\def\usecdclayout{0}

\documentclass{article}
\usepackage{arxiv}
\usepackage{graphics}
\usepackage{epsfig}
\usepackage{times}
\usepackage{amsmath}
\usepackage{amssymb}
\usepackage{amsthm}
\usepackage{algpseudocode,algorithm,algorithmicx}
\usepackage{mathsymbols}

\title{Exploiting linear substructure in LRKFs (Extended)}
\author{Marcus Greiff, Anders Robertsson and Karl Berntorp \\ \texttt{[marcus.greiff, anders.robersson]@control.lth.se}, \texttt{berntorp@merl.com}}

\begin{document}

\maketitle

\begin{abstract}
We exploit knowledge of linear substructure in the linear-regression Kalman filters (LRKFs) to simplify the problem of moment matching. The theoretical results yield quantifiable and significant computational speedups at no cost of estimation accuracy, assuming partially linear estimation models. The results apply to any symmetrical LRKF, and reductions in computational complexity are stated as a function of the cubature rule, the number of linear and nonlinear states in the estimation model respectively. The implications for the filtering problem are illustrated by numerical examples.
\end{abstract}

\thispagestyle{empty}
\pagestyle{empty}

\section{Introduction}\label{sec:introduction}
In this paper, we explore the incorporation of known linear substructure in the linear-regression Kalman filters (LRKFs) summarised in~\cite{steinbring2014lrkf,kurz2017linear}. For future reference, we refer to such filters as partially linear LRKFs, or PL-LRKFs for short. Many physical systems arising from Newtonian mechanics have some partially linear substructure in their dynamical equations, and a great number of systems have a partially linear measurement model. For state estimation with such systems, the large family of particle filters (PFs)~\cite{gordon1993novel,smith2013sequential} quickly become computationally intractable with the number of states that are to be estimated. This has historically been a major motivation for the development of Rao-Blackwellized particle filters (RBPFs)~\cite{schon2005marginalized,schon2006marginalized}, which assume a particle distribution in the nonlinear states and a Gaussian distribution in the linear states. However, due to stemming from the particle filtering framework, such filters also tend to be computationally cumbersome for large numbers of states, which is why a vast majority of nonlinear estimation applications still employ Gaussian approximate density filters (ADFs), such as the extended Kalman filter (EKF)~\cite{sarkka2013bayesian}.

For problems where the estimate distribution is likely to be uni-modal or the RBPF-variants are deemed computationally intractable, alternatives to the EKFs include the LRKFs. These are also Gaussian ADFs, but use various cubature rules in order to evaluate a set of moment integrals, instead of approximating the nonlinear functions by Taylor expansions, as done in the EKF. By using high-order cubature rules, these filters are often favored over the EKF for their estimation accuracy (see e.g.,~\cite{steinbring2014lrkf,kurz2017linear}), but much like the PFs, the LRKFs generally do not exploit linear substructure in the nonlinear estimation models.

Consequently, we analyze the problem of moment matching, that is, computing the first two moments of the joint distribution of input and output of a nonlinear function, for partially linear functions using various cubature rules. Specifically, we focus on the spherical cubature rule (SC) used in the cubature Kalman filter (CKF)~\cite{arasaratnam2009cubature,jia2013high}; the unscented transform (UT) used in the unscented Kalman filter (UKF) in~\cite{julier1997new,wan2001unscented}; the Gauss-Hermite cubature rule (GHC) used in the GHKF in~\cite{sarkka2013bayesian}; and the stochastic integration rule (SIR) used in the randomized unscented Kalman filter (RUKF) in~\cite{dunik2011development, dunik2013stochastic}. However, the results apply to all symmetric LRKFs of which the aforementioned filters are but a subset. As such, this work distinguishes itself from the relevant prior work in~\cite{morelande2007unscented,beutler2009gaussian} in two main respects. Firstly, in its generality: we are considering all symmetric LRKFs in terms of the cubature point set, and not specific cubature rules such as the RB-UKF in~\cite{morelande2007unscented}. Secondly, in that we are not using the conditionally linear structure~\cite{beutler2009gaussian}, but a partially linear structure in the otherwise nonlinear equations.

\newpage\subsection{Contributions}
In adition to the numerical results, the main theoretical contributions of this paper are as follows:
\begin{enumerate}
    \item[(i)] An expression for the moments of the joint distribution of input and output to a nonlinear function with a linear substructure, given a generic cubature rule defined by a set of integration points and weights.
    \item[(ii)] Conditions under which the joint distribution can be evaluated with a number of integration points that scale with the number of nonlinear states instead of with the total number of states.
    \item[(iii)] Conditions under which the square root factorization in the LRKFs only needs to be computed partially.
    \item[(iv)] Proofs that the conditions in (ii) and (iii) hold for the cubature rules used in the CKF, UKF and GHKF.
\end{enumerate}

\subsection{Overview}
We start by a brief review of the LRKFs in Section~\ref{sec:prelim}, followed by the definition of the moment matching problem in Section~\ref{sec:momentmathcing}. The main results are given in Section~\ref{sec:results}, and their implications for the filtering problem are subsequently illustrated in Section~\ref{sec:filtering}. Finally, numerical results are given in~\ref{sec:numerical}, and Section~\ref{sec:conclusion} closes the paper.

\subsection{Notation}
In the following, we let $\evec_i^N\in  \mathbb{R}^N$ denote a unit vector with the $i^{th}$ element set to 1, and all other elements set to zero. The vector $\mathbf{1}_N\in\mathbb{R}^N$ is a column vector of ones, the matrix $\Z_{N\times M}\in\mathbb{R}^{N\times M}$ is a zero matrix, and $\I_N\in\mathbb{R}^{N\times N}$ is the identity matrix. Vectors are denoted by bold font $\xvec$, and sets with calligraphic font $\Scal$ with $|\Scal|$ denoting set cardinality. Here, $\Ncal(\xvec|\mx, \Pxx)$ denotes a Gaussian probability density function over $\xvec$ with mean $\mx$ and covariance $\Pxx$, closely following the notation in~\cite{sarkka2013bayesian}. The sub-indexation $(\cdot)_k$ indicates a variable at a time-step $k$, and the notation $(\cdot)_{a|b}$ indicates a variable at a time $k=a$ conditioned on information up until and including $k=b$. Finally, we take $\otimes$ to denote the usual Kronecker product, and let $\star$ denote a redundant entry in a symmetric matrix.

\section{Preliminaries}\label{sec:prelim}
We consider a discrete-time systems on the form,
\begin{subequations}\label{eq:dynsys}
\begin{align}
\xvec_{k+1} &= \F(\xvec_{k},\qvec_{k})\hspace{2.1pt}\in\mathbb{R}^X,\label{eq:dyn}\\
\yvec_{k} &= \HH(\xvec_{k},\rvec_{k})\in\mathbb{R}^Y,\label{eq:meas}
\end{align}
\end{subequations}
where the process noise and measurement noise are Gaussian distributed, with $\qvec_{k}\sim\Ncal(\boldsymbol{0}, \Q_k)$, $\rvec_{k}\sim\Ncal(\boldsymbol{0}, \R_{k})$, and $\boldsymbol{0}\prec \Q_{k} = \Q_{k}^{\top}, \boldsymbol{0}\prec \R_{k} = \R_{k}^{\top}$. The objective is to recursively estimate the state $\xvec_k$ given $\yvec_{0:k}$ using the LRKFs. In the generic Gaussian approximate density filters (ADFs), of which the LRKFs are a subset, the distribution of the state-estimate at time step $k-1$ is approximated by a Gaussian,
\begin{equation*}
p({\xvec}_{k-1}|\yvec_{0:k-1})\hspace{-1pt}\approx\hspace{-1pt}\Ncal({\xvec}_{k-1}| \mx_{k-1}, \Pxx_{k-1}).
\end{equation*}
Given this approximation, the estimate distribution is propagated through the dynamics in~\eqref{eq:dyn}, yielding a prediction
\begin{equation}\label{eq:dynupdate}
p({\xvec}_k|\yvec_{0:k-1}) \approx \Ncal({\xvec}_{k|k-1}| \mx_{k|k-1}, \Pxx_{k|k-1}),
\end{equation}
and the joint distribution of the predicted state and measurement is approximated based on~\eqref{eq:meas}, as
\begin{equation}\label{eq:jointdist}
\Ncal\Bigg(
\begin{bmatrix}{\xvec}_{k|k-1}\\ \yvec_{k|k-1} \end{bmatrix}\Big|
\begin{bmatrix}\mx_{k|k-1}\\ \my_{k|k-1} \end{bmatrix}, \begin{bmatrix}
\Pxx_{k|k-1} & \Pxy_{k|k-1}\\
\Pyx_{k|k-1} & \Pyy_{k|k-1}
\end{bmatrix}\Bigg).
\end{equation}
Upon receiving a measurement ${\yvec}_{k|k-1} = \yvec_k$, we evaluate
\begin{equation*}
p({\xvec}_k|\yvec_{0:k})=\mathcal{N}({\xvec}_k | \mx_{k|k}, \Pxx_{k|k}),
\end{equation*}
using Bayes' rule, where for the multivariate Gaussian case,
\begin{flalign}\label{eq:conditional}
\mx_{k|k} &=\mx_{k|k-1} + \Pxy_{k|k-1}(\Pyy_{k|k-1})^{-1}(\yvec_k - \my_{k|k-1}),\nonumber\\
\Pxx_{k|k} &=\Pxx_{k|k-1} - \Pxy_{k|k-1}(\Pyy_{k|k-1})^{-1}\Pyx_{k|k-1}.
\end{flalign}

The difference among all of the Gaussian ADFs lies in the way the prediction in~\eqref{eq:dynupdate} and the joint distribution in~\eqref{eq:jointdist} are approximated. In the event of linear flow and measurement equations~\eqref{eq:dynsys}, the state-distribution will be Gaussian at all times, and the conditional distribution can be computed exactly. The above equations in~\eqref{eq:dynupdate},~\eqref{eq:jointdist}, and~\eqref{eq:conditional} then result in the familiar Kalman filter. However, if~\eqref{eq:dynsys} is nonlinear, many forms of approximations can be considered in~\eqref{eq:dynupdate} and~\eqref{eq:jointdist}. Here, direct approximation of the associated moment integrals, results in the large family of the linear regression Kalman filters (LRKFs)~\cite{steinbring2014lrkf}. This is commonly referred to as moment matching, and we will now analyze how knowledge of linear substructure in~\eqref{eq:dynsys} can be leveraged to simplify the numerical evaluation of these moment integrals in the context of the aforementioned SC, UT, GHC and SIR schemes.
%This includes various Gaussian cubature rules, of which the more common are the spherical-radial cubature rules (SRC), resulting in the Cubature Kalman filter (CKF)~\cite{arasaratnam2009cubature}, the unscented transform (UT) used to define the uncsented Kalman filter (UKF)~\cite{julier1997new}, the Gauss-Hermite cubature rules (GHC) resulting in the GHKF~\cite{ito2000gaussian}, and the MC-based approach in the randomized Unscented Kalman filter (RUKF)~\cite{dunik2011development}.

\section{The problem of moment matching}\label{sec:momentmathcing}
In the following, we simplify the notation by temporarily dropping the time indexation and only considering the computation of the first and second moments of the joint distribution of input and output to a nonlinear function with linear substructure. To investigate potential simplifications in the moment matching, consider a function $\fullnonlinfunc\;:\mathbb{R}^X\rightarrow\mathbb{R}^Y$, operating on a state $\xvec\in\mathbb{R}^X$, which can be partitioned into a linear part $\lvec\in\mathbb{R}^L$ and a nonlinear part $\zvec\in\mathbb{R}^Z$,
\begin{equation}\label{eq:distX}
\removeinpaper{
\xvec = \begin{bmatrix}
\zvec\\ \lvec
\end{bmatrix}\in\mathbb{R}^X = \mathbb{R}^{Z + L},\qquad}
\mathcal{N}(\xvec|\mx,\Pxx) = 
\mathcal{N}\Bigg(
\begin{bmatrix}
\zvec\\\lvec
\end{bmatrix}\Bigg|
\begin{bmatrix}
\mn\\\ml
\end{bmatrix},
\begin{bmatrix}
\Pzz&\Pzl\\\Pxz&\Pll
\end{bmatrix}
\Bigg),
\end{equation}
where the $\fullnonlinfunc$ has some linear substructure on the form
\begin{equation}\label{eq:func}
\yvec = \G(\xvec)\triangleq
\begin{bmatrix}
\nonlinfunc(\zvec)\\
\A\xvec
\end{bmatrix}.
\end{equation}

In essence, only knowing the structure in~\eqref{eq:distX} and ~\eqref{eq:func}, we seek a simplified moment-matching of the joint density
\begin{equation}
\mathcal{N}\Bigg(
\begin{bmatrix}
\xvec\\\yvec
\end{bmatrix}\Bigg|
\begin{bmatrix}
\mx\\\my
\end{bmatrix},
\begin{bmatrix}
\Pxx&\Pxy\\\Pyx&\Pyy
\end{bmatrix}
\Bigg),
\end{equation}
by approximate evaluation of the moment integrals, cf.~\cite{ito2000gaussian},
\begin{subequations}\label{eq:momentintegrals}
\begin{align}
\my &=\int_{\mathbb{R}^X}\fullnonlinfunc(\xvec)\mathcal{N}(\xvec|\mx,\Pxx)d\xvec,\\
\Pxy&=\int_{\mathbb{R}^X}(\xvec -\mx)(\fullnonlinfunc(\xvec) -\my)^{\top}\mathcal{N}(\xvec|\mx,\Pxx)d\xvec,\\
\Pyy&=\int_{\mathbb{R}^X}(\fullnonlinfunc(\xvec)-\my)(\fullnonlinfunc(\xvec)-\my)^{\top}\mathcal{N}(\xvec|\mx,\Pxx)d\xvec.
\end{align}
\end{subequations}
The question is how to leverage knowledge of the linear substructure in~\eqref{eq:distX} when evaluating~\eqref{eq:momentintegrals}, and what implications this has for the resulting filtering problem in Section~\ref{sec:prelim}.

\subsection{Approximate moment matching}
Moving forward, we assume a non-degenerate distribution over $\xvec$, with $\Pxx ={\Pxx}^{\top}\succ \boldsymbol{0}$, such that there exists
\begin{equation}\label{eq:lowertriangulardecomp}
\Pxx = \Lxx{\Lxx}^{\top}, \qquad
\Lxx =
\begin{bmatrix}
\Lnn & \Z\\
\Lln & \Lll
\end{bmatrix},
\end{equation}
where $\Lxx, \Lnn, \Lll$ are all lower-triangular. Given this, there are many ways of approximating the moment integrals. However, all of the prior mentioned schemes first take a coordinate transform $\xivec = (\Lxx)^{-1} (\xvec - \mx)$, and then approximate the moment integrals in a finite sum
\begin{subequations}\label{eq:approx}
\begin{align}
\int_{\mathbb{R}^X}\fullnonlinfunc(\xvec)\mathcal{N}(\xvec|\mx,\Pxx)d\xvec &=\int_{\mathbb{R}^X}\fullnonlinfunc(\mx + \Lxx\xivec)\mathcal{N}(\xivec|\Z,\I)d\xivec\\
&= \int_{\mathbb{R}}\cdots\int_{\mathbb{R}}\fullnonlinfunc(\mx + \Lxx\xivec)\mathcal{N}(\xi_1|0,1)d\xi_1\cdots\mathcal{N}(\xi_X|0,1) d\xi_X\\
&\approx \sum\limits_{i = 1}^{C(X)}w^{(i)}\fullnonlinfunc (\mx+\Lxx\xivec^{(i)}),
\end{align}
\end{subequations}
by a total of $C(X)$ pairs of weights and integration points $\mathcal{P} = \{(w^{(i)},\xivec^{(i)})\}_{i =1}^{C(X)}$, where the cardinality $C(X) = |\Pcal|$ increases with $X = \text{dim}(\xvec)$. In this notation, the LRKF moment approximations can be written on the form
\begin{subequations}\label{eq:approxmomentint}
\begin{align}
\mathcal{X}^{(i)} &= \mx + \Lxx\xivec^{(i)}
\hspace{10pt}
i\in\{1,...,|\mathcal{P}|\}, \\
\mathcal{Y}^{(i)} &= \fullnonlinfunc(\mathcal{X}^{(i)})
\hspace{37pt}
i\in\{1,...,|\mathcal{P}|\}, \\
\my &
\removeinpaper{=\int_{\mathbb{R}^X}\gvec(\xvec)\mathcal{N}(\xvec|\mx,\Pxx)d\xvec \hspace{97pt}}
\approx\sum\limits_{i = 1}^{|\mathcal{P}|}w^{(i)}\mathcal{Y}^{(i)},\\
\Pxy&
\removeinpaper{=\int_{\mathbb{R}^X}(\xvec -\mx)(\gvec(\xvec) -\my)^{\top}\mathcal{N}(\xvec|\mx,\Pxx)d\xvec \hspace{13pt}}
\approx\sum\limits_{i = 1}^{|\mathcal{P}|}w^{(i)}(\mathcal{X}^{(i)} -\mx)(\mathcal{Y}^{(i)} -\my)^{\top},\\
\Pyy&
\removeinpaper{=\int_{\mathbb{R}^X}(\gvec(\xvec) -\my)(\gvec(\xvec) -\my)^{\top}\mathcal{N}(\xvec|\mx,\Pxx)d\xvec}
\approx\sum\limits_{i = 1}^{|\mathcal{P}|}{w}^{(i)}(\mathcal{Y}^{(i)} -\my)(\mathcal{Y}^{(i)} -\my)^{\top}.
\end{align}
\end{subequations}

To proceed with the analysis, we start with some additional definitions. Let $\mathbf{N} = \begin{bmatrix}
\I_{Z}& \Z_{Z\times L}
\end{bmatrix}$ and  $\bar{\mathbf{N}} = \begin{bmatrix}
\Z_{L\times Z}& \I_{L}
\end{bmatrix}$, such that $\mathbf{N}\xvec = \zvec$ and $\bar{\mathbf{N}}\xvec = \lvec$. For simplicity, in analyzing the cubature rules in the context of the linear substructure in~\eqref{eq:func}, we further categorize $\Pcal$ into three categories; central ($c$), linear ($l$), and nonlinear ($z$). The points where $\xivec = \Z$, we refer to as central points with a sub-index $(\cdot)_c$; if $\xivec$ is at the origin in the dimensions corresponding to the input of the nonlinear function $\nonlinfunc$ used to define $\fullnonlinfunc$ in~\eqref{eq:func}, we refer to these points as the linear points with a sub-index $(\cdot)_l$, as they only differ from the origin in the linear dimensions of the state $\lvec\subseteq\xvec\in\mathbb{R}^X$; if $\xivec$ differs from the origin in the input to $\nonlinfunc$, the points and weights are sub-indexed $(\cdot)_z$. \replaceinpaper{In mathematical terms,
\begin{align}\label{eq:Pdef}
\mathcal{P}_c &= \{(w,\xivec)\in\mathcal{P}| \xivec = \Z\},&
\mathcal{P}_l &= \{(w,\xivec)\in\mathcal{P}\backslash \Pcal_c| \mathbf{N}\xivec = \Z\},&
\mathcal{P}_z &= \mathcal{P}\backslash (\mathcal{P}_c\cup \mathcal{P}_l)
\end{align}
This allows us to categorize the set of points and weights into three subsets
\begin{subequations}\label{eq:subsetdef}
\begin{align}
w_c^{(i)} &= w^{(i)}\in\mathbb{R}, &\Xi_c^{(i)} &= \xivec^{(i)}\in\mathbb{R}^X, &\mathcal{X}_c^{(i)} &= \mx + \Lxx\xivec^{(i)}\in\mathbb{R}^X,&\forall&(w^{(i)},\xivec^{(i)})\in\mathcal{P}_c,\\
w_z^{(i)} &= w^{(i)}\in\mathbb{R}, &\Xi_z^{(i)} &= \xivec^{(i)}\in\mathbb{R}^X, &\mathcal{X}_z^{(i)} &= \mx + \Lxx\xivec^{(i)}\in\mathbb{R}^X,&\forall&(w^{(i)},\xivec^{(i)})\in\mathcal{P}_z,\\
w_l^{(i)} &= w^{(i)}\in\mathbb{R}, &\Xi_l^{(i)} &= \xivec^{(i)}\in\mathbb{R}^X, &\mathcal{X}_l^{(i)} &= \mx + \Lxx\xivec^{(i)}\in\mathbb{R}^X,&\forall&(w^{(i)},\xivec^{(i)})\in\mathcal{P}_l,
\end{align}
\end{subequations}
and, a number of associated matrices
\begin{subequations}\label{eq:matdef}
\begin{align}
\wvec_c &=\begin{bmatrix}
w_c^{(1)} & \cdots & w_c^{(|\Pcal_c|)}
\end{bmatrix},&
\Xicalbf_c &=\begin{bmatrix}
\Xi_c^{(1)} & \cdots & \Xi_c^{(|\Pcal_c|)}
\end{bmatrix},&
\Xcalbf_c &=\begin{bmatrix}
\mathcal{X}_c^{(1)} & \cdots & \mathcal{X}_c^{(|\Pcal_c|)}
\end{bmatrix}\\
\wvec_z &=\begin{bmatrix}
w_z^{(1)} & \cdots & w_z^{(|\Pcal_z|)}
\end{bmatrix},&
\Xicalbf_z &=\begin{bmatrix}
\Xi_z^{(1)} & \cdots & \Xi_z^{(|\Pcal_z|)}
\end{bmatrix},&
\Xcalbf_z &=\begin{bmatrix}
\Xcal_z^{(1)} & \cdots & \Xcal_z^{(|\Pcal_z|)}
\end{bmatrix}\\
\wvec_l &=\begin{bmatrix}
w_l^{(1)} & \cdots & w_l^{(|\Pcal_l|)}
\end{bmatrix},&
\Xicalbf_l &=\begin{bmatrix}
\Xi_l^{(1)} & \cdots & \Xi_l^{(|\Pcal_l|)}
\end{bmatrix},&
\Xcalbf_l &=\begin{bmatrix}
\mathcal{X}_l^{(1)} & \cdots & \mathcal{X}_l^{(|\Pcal_l|)}
\end{bmatrix}
\end{align}
\end{subequations}
and
\begin{align}\label{eq:fulldef}
\wvec = 
\begin{bmatrix}
\wvec_c&\wvec_z&\wvec_l
\end{bmatrix},\;\;
\Xicalbf = 
\begin{bmatrix}
\Xicalbf_c&\Xicalbf_z&\Xicalbf_l
\end{bmatrix},\;\;
\Xcalbf = 
\begin{bmatrix}
\Xcalbf_c&\Xcalbf_z&\Xcalbf_l
\end{bmatrix},\;\;
\Wcalbf = \text{diag}(\wvec),\;\;
\Wcalbf_z = \text{diag}(\wvec_z)
\end{align}
with
\begin{align}\label{eq:Zdef}
\mathcal{Z}^{(i)} &= \mathbf{N}
\mathcal{X}^{(i)},&
\mathcal{Z}_c^{(i)} &= \mathbf{N}
\mathcal{X}_c^{(i)},&
\mathcal{Z}_z^{(i)} &= \mathbf{N}
\mathcal{X}_z^{(i)}&
\mathcal{Z}_l^{(i)} &= \mathbf{N}
\mathcal{X}_l^{(i)}.
\end{align}
}{
In the following, we let $a$ denote a sub-index $c, z,$ or $l$, and let
\begin{subequations}\label{eq:Pdef}
\begin{align}
\mathcal{P}_c &= \{(w,\xivec)\in\mathcal{P}| \xivec = \Z\},\\
\mathcal{P}_l &= \{(w,\xivec)\in\mathcal{P}\backslash \Pcal_c| \mathbf{N}\xivec = \Z\},\\
\mathcal{P}_z &= \mathcal{P}\cap (\mathcal{P}_c\cup \mathcal{P}_l).
\end{align}
\end{subequations}
Here, for any of the point sets $\Pcal_a$, we let
\begin{align}\label{eq:subsetdef}
\begin{cases}
w_a^{(i)} = w^{(i)},\\
\Xi_a^{(i)}= \xivec^{(i)},\\
\mathcal{X}_a^{(i)}=\mx+ \Lxx\xivec^{(i)}\\
\Zcal_a^{(i)} = \N\mathcal{X}_a^{(i)}
\end{cases}
\forall&(w^{(i)},\xivec^{(i)})\in\mathcal{P}_a,
\end{align}
and define associated matrices
\begin{subequations}\label{eq:matdef}
\begin{align}
\wvec_a &=\begin{bmatrix}
w_a^{(1)} & \cdots & w_a^{(|\Pcal_a|)}
\end{bmatrix}\in\mathbb{R}^{1\times |\Pcal_a|},\\
\Xicalbf_a &=\begin{bmatrix}
\Xi_a^{(1)} & \cdots & \Xi_a^{(|\Pcal_a|)}
\end{bmatrix}\in\mathbb{R}^{X\times |\Pcal_a|},\\
\Xcalbf_a &=\begin{bmatrix}
\mathcal{X}_a^{(1)} & \cdots & \mathcal{X}_a^{(|\Pcal_a|)}
\end{bmatrix}\in\mathbb{R}^{X\times |\Pcal_a|},\\
\Zcalbf_a &=\begin{bmatrix}
\mathcal{Z}_a^{(1)} & \cdots & \mathcal{Z}_a^{(|\Pcal_a|)}
\end{bmatrix}\in\mathbb{R}^{Z\times |\Pcal_a|}.
\end{align}
\end{subequations}
We also define the combined matrices,
\begin{subequations}\label{eq:fulldef}
\begin{align}
\wvec &= 
\begin{bmatrix}
\wvec_c&\wvec_z&\wvec_l
\end{bmatrix}, & \Wcalbf &= \text{diag}(\wvec),\\
\Xicalbf &= 
\begin{bmatrix}
\Xicalbf_c&\Xicalbf_z&\Xicalbf_l
\end{bmatrix}, & \Wcalbf_z &= \text{diag}(\wvec_z),\\
\Xcalbf &= 
\begin{bmatrix}
\Xcalbf_c&\Xcalbf_z&\Xcalbf_l
\end{bmatrix},& \Zcalbf &= 
\begin{bmatrix}
\Zcalbf_c&\Zcalbf_z &\Zcalbf_l
\end{bmatrix}.
\end{align}
\end{subequations}}
These definitions are illustrated in Figure~\ref{fig:definitions}, where a set of 5 integration points are plotted along the last nonlinear-state dimension and the first linear state dimension. Here, we note that by virtue of the lower triangular decomposition in~\eqref{eq:lowertriangulardecomp}, the points computed $\Xcal_l^{(1)}, \Xcal_l^{(2)}$ do not deviate from the central point $\Xcal_c^{(1)}$ in the part of the state-space that serves as the input to the nonlinear function $\gvec$ defining $\fullnonlinfunc$. As a result, the sum of the output points $\Ycal^{(4)}$ and $\Ycal^{(5)}$ is equal to $\Ycal^{(1)}$, and this is precisely what we will exploit in the following section. Again, this is made possible by the choice of the lower triangular decomposition. Taking a dense symmetric decomposition used in defining the change of variables in~\eqref{eq:approx} implies that we need to evaluate the nonlinear function $\gvec$ in five different points along $\evec_Z^X$ (green in Figure~\ref{fig:definitions}), whereas using a triangular decomposition requires evaluation in $\gvec$ in three points.

\begin{figure*}[t]
    \centering
    \includegraphics[width=\textwidth]{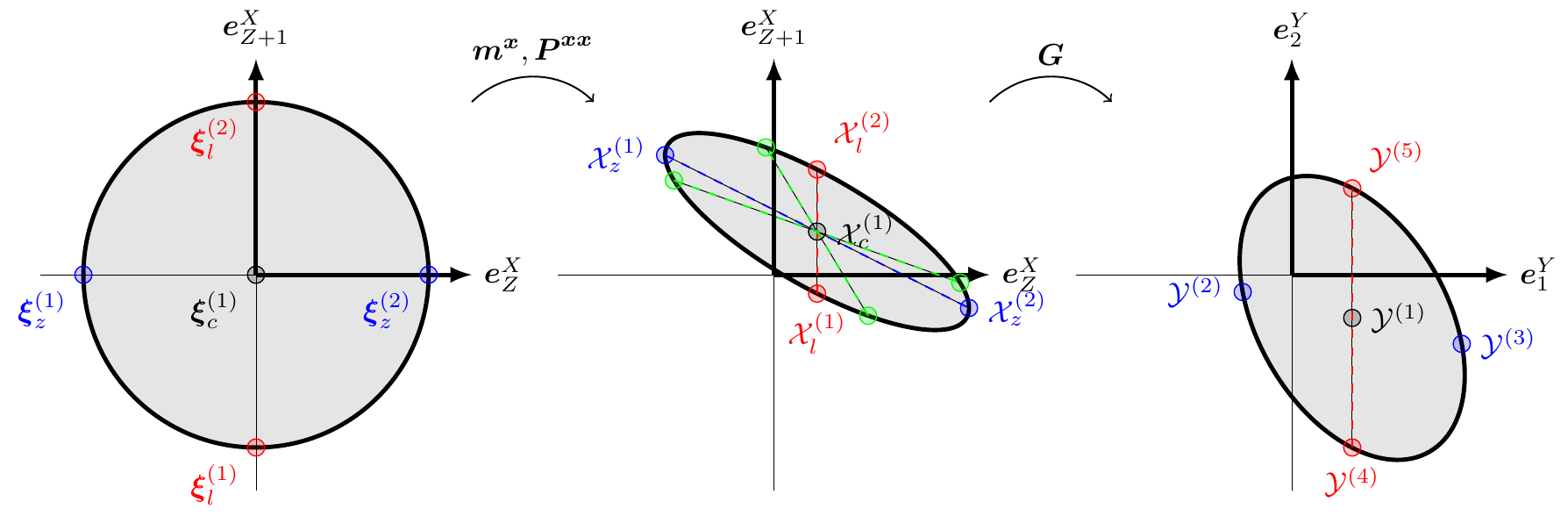}
    \caption{Illustration of the set $\Xicalbf$ partitioned into the sets $\Pcal_c, \Pcal_z, \Pcal_l$ for the unscented transform (left). Resulting integration points $\Xcal$, here computed with a Cholesky decomposition $\Lxx$ (blue/red), and also with a symmetric square root factorization analogous to Matlab's \texttt{sqrtm()} function (green).}
    \label{fig:definitions}
\end{figure*}

\newpage\section{Main result}\label{sec:results}
We now state some basic properties common to the prior mentioned cubature rules, which will be verified for the SC, UT, GHC, but certainly encompassing more of the LRKFs. We will then proceed to use these properties in the simplification of the evaluation of the joint distributions.

\begin{assumption}\label{rem:symmetry}
Symmetry, such that $\forall(w^{(i)},\xivec^{(i)})\in\Pcal_z\cup\Pcal_l$, $\exists(w^{(j)},\xivec^{(j)})\in\Pcal_z\cup\Pcal_l$ s.t. $(w^{(i)},\xivec^{(i)}) = (w^{(j)},-\xivec^{(j)})$.
\end{assumption}

\begin{assumption}\label{rem:productFirst}
Consistency in the first moment for linear maps, with $\sum_{(w,\xivec)\in\Pcal} w = 1$.
\end{assumption}

\begin{assumption}\label{rem:product}
Consistency in the second moment for linear maps, with $\Xicalbf\Wcalbf\Xicalbf^{\top} =\I$.
\end{assumption}

\begin{rem}
The second and third assumption should hold for any conventional LRKF, and will be demonstrated explicitly for the integration point sets defining the considered cubature rules. The first assumption holds for any of the more common LRKFs, and can be enforced as a constraint in the optimization in the Smart Sampling Kalman filter (S$^2$KF) in~\cite{steinbring2014lrkf}). The asusmptions can also be shown to hold for the SIR in the RUKF, but here omitted for brevity.
\end{rem}

\subsection{The spherical cubature rule}\label{sec:CKF}
The spherical cubature rule used in the CKF, originally presented in~\cite{arasaratnam2009cubature}, scales as $|\Pcal| = C(X) = 2X$, and is defined by a set of points where $\Pcal_c = \emptyset$ implying that $\Pcal = \Pcal_z\cup\Pcal_l$, and the weights and integration points are given by
\begin{equation*}
\wvec = (2X)^{-1}\mathbf{1}_{2X}^{\top}, \qquad \Xicalbf = \sqrt{X}
\begin{bmatrix}
\I&-\I
\end{bmatrix}.
\end{equation*}
\begin{rem}
Assumption~\ref{rem:symmetry} clearly holds. As $\sum_{(w,\xivec)\in\Pcal} w = 2X(2X)^{-1} = 1$, and $\Xicalbf\Wcalbf\Xicalbf^{\top} = (2X)^{-1}(X\I + X\I) = \I$, showing that both Assumptions~\ref{rem:productFirst} and~\ref{rem:product} also hold. Note that this cubature rule has no central point, with $\Pcal_c = \emptyset$.
\end{rem}

\subsection{The unscented transform}\label{sec:UKF}
The unscented transform, used to define the celebrated UKF filter originally presented in~\cite{julier1997new}, scales as $|\Pcal| = C(X) = 2X + 1$. This transform is determined by a length scale parameter $\lambda = \alpha^2(X + \kappa) - X$, for some constant $\alpha,\kappa > 0$. The UT has one central integration point, with $\Pcal_c = (w_c^{(1)},\xivec_c^{(1)})=(\lambda(\lambda + X)^{-1},\Z)$, and $\Pcal$ is given by
\begin{equation*}
\wvec \hspace{-2pt}=\hspace{-2pt} \frac{1}{\lambda + X}\begin{bmatrix}
\lambda & \frac{1}{2}\mathbf{1}_{2X}^{\top}
\end{bmatrix}, \;\; \Xicalbf\hspace{-2pt}=\hspace{-2pt}
\sqrt{(\lambda \hspace{-2pt}+\hspace{-2pt} X)}\begin{bmatrix}
\Z_{X\times 1}\hspace{-2pt}&\hspace{-2pt}\I_X\hspace{-2pt}&\hspace{-2pt}-\I_X
\end{bmatrix}
\end{equation*}

\begin{rem}
For every point in $\Pcal_z\cup\Pcal_l$, Assumption~\ref{rem:symmetry} holds. Furthermore, $\sum_{(w,\xivec)\in\Pcal} w = \lambda(\lambda + X)^{-1} + 2X(2(\lambda + X))^{-1} = (\lambda + X)(\lambda + X)^{-1} = 1$, thus Assumption~\ref{rem:productFirst} holds. Finally, as $\Xicalbf\Wcalbf\Xicalbf^{\top} = \lambda(\lambda + X)^{-1}\Z_{X\times X} +(2(\lambda + X))^{-1}((\lambda + X)\I_X + (\lambda + X)\I_X) = \I_X$, Assumption~\ref{rem:product} holds. 
\end{rem}

\subsection{The Gauss-Hermite cubature rule}\label{sec:GHKF}
Another form of cubature is the GHC, to our knowledge first presented in~\cite{golub1969calculation} and used to construct the first GHKF in~\cite{ito2000gaussian}, comprehensively summarized in~\cite{sarkka2013bayesian}. Here, $H_p(x)$ denotes the probabilistic Hermite polynomial of order $p$, as
\begin{equation}
H_0\hspace{-1.5pt}=\hspace{-1.5pt}1, \;\; H_1\hspace{-1.5pt}=\hspace{-1.5pt}x, \;\; H_{p+1}(x)\hspace{-1.5pt}=\hspace{-1.5pt}xH_{p}(x)\hspace{-1.5pt}-\hspace{-1.5pt}pH_{p-1}(x),
\end{equation}
and $r_p^i$ denotes the $i^{th}$ root of $H_{p}(x)$. For each root, we compute an associated number
\begin{equation}
\alpha_p^i = p!(pH_{p-1}(r_p^i)))^{-2}.
\end{equation}
In the one-dimensional quadrature, $r_p^i$ and $\alpha_p^i$ would form the set of integration points and weights. However, this can easily be extended to a multi-dimensional cubature, then with the set of weights and points defined as
\begin{equation*}
\Pcal \hspace{-2pt}= \hspace{-2pt} \Bigg\{\Bigg(\prod_{j = 1}^X\alpha_p^{k_j}, \sum\limits_{j = 1}^Xr_p^{k_j}\evec_j^X\Bigg)|k_j  \hspace{-2pt}\in \hspace{-2pt} \{1,...,p\} \forall j  \hspace{-2pt}\in \hspace{-2pt} \{1,...,X\}\Bigg\}.
\end{equation*}
Thus, the number of integration points scales exponentially with the dimension of the state, with $C(X) = p^X$.
\begin{rem}
The Hermite polynomials are symmetric, that is, for every $r_p^i\neq 0$, $\exists r_p^j$ such that $r_p^i = -r_p^j$ and since $H_p(x) = (-1)^pH_p(-x)$, we have that $\alpha_p^i = \alpha_p^j$ for each such pair of roots. Thus, for every point in $\Pcal_z\cup\Pcal_l$, Assumption~\ref{rem:symmetry} clearly holds. To see that Assumption~\ref{rem:product} is satisfied, let $\gvec(\xvec) = \xvec$ and take $\mx = \Z$ and $\Pxx = \I$. In this case, $\Pyy = \I$, and, since the $p^{th}$-order Gauss-Hermite cubature rule is exact for polynomials of order $p$~\cite{ito2000gaussian}, and the second moment in this case is a second order polynomial in $\xvec$ with a Gaussian weight function, we have that
\replaceinpaper{
\begin{equation}
\I = \Pyy = \int_{\mathbb{R}^X}(\I\xvec - \mx)(\I\xvec - \mx)^{\top}\mathcal{N}(\xvec|\Z,\I)\text{d}\xvec = \sum\limits_{i = 1}^{|\Pcal|}w^{(i)}\xivec^{(i)}(\xivec^{(i)})^{\top}
 =  \Xicalbf\Wcalbf\Xicalbf^{\top},
\end{equation}
}{
$\Xicalbf\Wcalbf\Xicalbf^{\top} = \I$
}
for any GHC of order $p \geq 2$. Similarly, Assumption~\ref{rem:productFirst} holds for all $p \geq 1$.
\end{rem}

\subsection{Exploiting the linear substructure}
We will now attempt to exploit the linear substructure in~\eqref{eq:func} under the assumption that conditions~\ref{rem:symmetry},~\ref{rem:productFirst} and~\ref{rem:product} all hold. The key idea here is to partition the state as done in~\eqref{eq:distX}, and then use the lower-triangular Cholesky decomposition in the change of variables in~\eqref{eq:approx}, and subsequent evaluation of the set $\Xcal$, instead of a symmetric square root factorization, otherwise commonly used in the LRKFs. This special structure ensures that for any element in the set $\Pcal_c\cup\Pcal_l$, we have that $\mathbf{N}\Ycal^{(i)} = \mathbf{N}\fullnonlinfunc(\Xcal^{(i)}) = \nonlinfunc(\Zcal^{(i)}) = \nonlinfunc(\mn)$. This property does not hold for the symmetric square root factorization,  as illustrated in Figure~\ref{fig:definitions}, and with it in mind, we can proceed by attacking the expressions in~\eqref{eq:approxmomentint} using Assumptions~\ref{rem:symmetry}, \ref{rem:productFirst}, and \ref{rem:product}. The result is stated in Proposition~\ref{thm:main} and the Corollaries~\ref{cor:compgains}-\ref{cor:otherstruct}, with proofs in the Appendix.

\begin{prop}\label{thm:main}
For the state in~\eqref{eq:distX}, its joint distribution with output of the structured nonlinear function in~\eqref{eq:func},
\begin{equation}
\mathcal{N}\Bigg(
\begin{bmatrix}
\xvec\\\yvec
\end{bmatrix}\Bigg|
\begin{bmatrix}
\mx\\\my
\end{bmatrix},
\begin{bmatrix}
\Pxx&\Pxy\\\Pyx&\Pyy
\end{bmatrix}
\Bigg),
\end{equation}
computed using a cubature rule defined by a set of points and weights $\Pcal = \{(w^{(i)},\xivec^{(i)})\}_{i = 1}^{C(X)}$ satisfying assumptions~\ref{rem:symmetry},~\ref{rem:productFirst}, and ~\ref{rem:product}, the moments of the joint distribution is given by
\replaceinpaper{
\begin{subequations}\label{eq:jointThmMain}
\begin{align}
\my &=\begin{bmatrix} w_{cl}\nonlinfunc(\mn)+
\Gcalbf_z\wvec_z^{\top}\\
\A\mx
\end{bmatrix},\\
{\Pxy}^{\top}&=
\begin{bmatrix}
\Gcalbf_z\Wcalbf_z\Xicalbf_z^{\top}{\Lxx}^{\top}\\
\A \Pxx
\end{bmatrix},\\
\Pyy&=
\begin{bmatrix}
(\Gcalbf_z+\Ccalbf_z)\Wcalbf_z(\Gcalbf_z+\Ccalbf_z)^{\top}+w_{cl}\uvec_{\lvec}\uvec_{\lvec}^{\top} &
\Gcalbf_z\Wcalbf_z\Xicalbf_z^{\top} {\Lxx}^{\top}\A^{\top}\\
\A \Lxx \Xicalbf_z\Wcalbf_z\Gcalbf_z^{\top} & 
\A \Pxx\A^{\top}
\end{bmatrix},
\end{align}
\end{subequations}
}{
\begin{align}\label{eq:jointThmMain}
\my &\hspace{-2pt}=\hspace{-2pt}\begin{bmatrix} w_{cl}\nonlinfunc(\mn)+
\Gcalbf_z\wvec_z^{\top}\\
\A\mx
\end{bmatrix},\\\nonumber
\Pxy&\hspace{-2pt}=\hspace{-2pt}
\begin{bmatrix}
\Lxx\Xicalbf_z\Wcalbf_z\Gcalbf_z^{\top} &  \Pxx\A^{\top}
\end{bmatrix},\\
\Pyy&\hspace{-2pt}=\hspace{-2pt}
\begin{bmatrix}
(\Gcalbf_z\hspace{-2pt}+\hspace{-2pt}\Ccalbf_z)\Wcalbf_z(\Gcalbf_z\hspace{-2pt}+\hspace{-2pt}\Ccalbf_z)^{\top}\hspace{-2pt}+\hspace{-2pt}w_{cl}\uvec_{\lvec}\uvec_{\lvec}^{\top} &
\star\\
\A \Lxx \Xicalbf_z\Wcalbf_z\Gcalbf_z^{\top} & 
\A \Pxx\A^{\top}
\end{bmatrix},\nonumber
\end{align}}
with $\Pcal_z,\wvec_z,\Wcalbf_z, \Xicalbf_z$ given in~\eqref{eq:Pdef}-\eqref{eq:fulldef}, and
\begin{subequations}
\begin{align}
w_{cl} &= 1 - \wvec_z \mathbf{1}_{|\Pcal_z|},\\
\Gcalbf_z &= \begin{bmatrix}
\nonlinfunc(\Zcal^{(1)}) & \cdots & \nonlinfunc(\Zcal^{(|\Pcal_z|)})
\end{bmatrix},\\
\uvec &= -w_{cl}\nonlinfunc(\mn) - \Gcalbf_z\wvec_z^{\top}\\
\uvec_{\lvec}&= \nonlinfunc(\mn) + \uvec.\\
\Ccalbf_z &= \uvec\mathbf{1}_{|\Pcal_z|}^{\top}.
\end{align}
\end{subequations}
\end{prop}

\begin{cor}\label{cor:compgains}
If, $\forall(w^{(i)},\xivec^{(i)})\in\Pcal_z\cup\Pcal_l$, $\exists(w^{(j)},\xivec^{(j)})\in\Pcal_z\cup\Pcal_l$ s.t. $(w^{(i)},\N\xivec^{(i)},\bar{\N}\xivec^{(i)}) = (w^{(j)},{\N}\xivec^{(j)},-\bar{\N}\xivec^{(j)})$, which is clearly satisfied in the SC, UT, and GHC, then Proposition~\ref{thm:main} only requires evaluation of the nonlinear function $\nonlinfunc$ in $C(Z)$ points instead of evaluating the function $\fullnonlinfunc$ in $C(X)$ points as done in the original cubature rules.
\end{cor}

\begin{rem}
The proof idea in Corollary~\ref{cor:compgains} consists of defining a point set where only unique columns of $\Xicalbf_z$ are considered, adjusting the associated set of weights accordingly, and resulting in a modified unique point set $\Pcal_z^u$. With this point set, and granted Assumptions 1-3, evaluation of Proposition 1 becomes identical when using $\Pcal_z$ and $\Pcal_z^u$.
\end{rem}

\begin{cor}\label{cor:A}
If $\bar{\mathbf{N}}\Xicalbf_z = \Z$, as is the case with the SC, UT, as well as the GHC (if considering its unique point set $\Pcal_z^u$), then the joint distribution~\eqref{eq:jointThmMain} in Proposition~\ref{thm:main} can be evaluated with a partial column-wise Cholesky decomposition of $\Pxx$. Only the first $Z$ columns of $\Lxx$ are needed.
\end{cor}

\begin{cor}\label{cor:otherstruct}
The results for the substructure in~\eqref{eq:distX} can be generalized to structured nonlinear functions on the form
\begin{equation}\label{eq:otherstruct}
\yvec = \begin{bmatrix}
\A_1\xvec + \gvec(\zvec)\\
\A_2\xvec
\end{bmatrix}.
\end{equation}
\end{cor}

\section{Implications for the filtering problem}\label{sec:filtering}
To illustrate the implications for the filtering problem, we start by defining the regular LRKF filters in the context of the notation in Section~\ref{sec:prelim}. Here, we will use a given cubature rule in approximating the moment integrals, both for the flow function $\F:\mathbb{R}^X\rightarrow \mathbb{R}^X$ and the measurement function $\HH:\mathbb{R}^X\rightarrow \mathbb{R}^Y$. As such, the linear and nonlinear states may not be the same for the two functions. Consequently, we assume that the functions have an associated orthogonal permutation matrix, here denoted $\T_{\F}$ and $\T_{\HH}$ respectively, such that
\begin{equation}
    \begin{bmatrix}\bar{\zvec} \\ \bar{\lvec}\end{bmatrix}=\bar{\xvec} = \T_{\F}\xvec,
\end{equation}
yields a permuted state vector where the $\bar{\zvec}$ states are nonlinear in $\F$, and the $\bar{\lvec}$ are linear in $\F$. If the state is Gaussian distributed with $\Ncal(\xvec|\mx, \Pxx)$, then this results in a change of both the mean and covariance, with
$\Ncal(\bar{\xvec}|\mxbar, \Pxxbar) = \Ncal(\bar{\xvec}|\T_{\F}\mx, \T_{\F}\Pxx(\T_{\F})^{\top})$. For future reference, we let the transformation of the moments of the state $\xvec$ by a transformation $\T_{\F}$ be denoted $\texttt{\{}\mxbar, \Pxxbar\texttt{\}} = \texttt{permute}(\T_{\F}, \mx, \Pxx)$. Similarly, we will have a different set of points, $\Pcal$, for the flow and measurement equations, and these are simply defined as $\Pcal_{\F}$ and $\Pcal_{\HH}$ respectively. Furthermore, we denote the full Cholesky factorization of a matrix by $\Lxx = \texttt{chol}(\Pxx)$, and the Cholesky-Crout algorithm computing the first $Z$ columns of the square-root factorization as $\texttt{\{}\Lnn, \Lnl\texttt{\}} = \texttt{chol\_crout}(\Pxx)$. The generic LRKF is given in Algorithm~\ref{alg:LRKF}, to be compared with the PL-LRKF in Algorithm~\ref{alg:LSLRKF}.

\begin{algorithm}\small
\caption{The generic LRKF}\label{alg:LRKF}
    \begin{algorithmic}[1]%\noindent 
		\State Initialize $\Ncal(\hat\xvec_{0}|\mxhat_{0}, \Pxxhat_{0}),\Pcal_{\F},\Pcal_{\HH}$
		
		\For{$k=1$ {\bf to} $K$}
		
		\textit{\textbf{// Time update}}
		\State $\Lxxhat_{k-1} = \texttt{chol}(\Pxxhat_{k-1})$
        \State Compute ${\Xcal}_{k-1}$ by~\eqref{eq:subsetdef} using $\mxhat_{k-1},\Lxxhat_{k-1},\Pcal_{\F}$
        \State Evaluate $\mxhat_{k|k-1}$ and $\Pxxhat_{k|k-1}$ using~\eqref{eq:approxmomentint} for~\eqref{eq:dyn}
        
        \textit{\textbf{// Measurement update}}
		\State $\Lxxhat_{k|k-1} = \texttt{chol}(\Pxxhat_{k|k-1})$
		\State Compute $\Xcal_{k|k-1}$ by~\eqref{eq:subsetdef} using $\mxhat_{k|k-1},\Lxxhat_{k|k-1}, \Pcal_{\HH}$
        \State \parbox[t]{\dimexpr .94\linewidth-\algorithmicindent}{Evaluate $\{\myhat_{k|k-1}, \Pyyhat_{k|k-1}, \Pxyhat_{k|k-1}\}$ using~\eqref{eq:approxmomentint} for the function in~\eqref{eq:meas} with the point set in $\mathcal{P}_{\HH}$}
        \State Evaluate $\{\mxhat_{k}, \Pxxhat_{k}\}$ using~\eqref{eq:conditional}
        \EndFor
	\end{algorithmic}
\end{algorithm}

\begin{algorithm}\small
\caption{The generic PL-LRKF}\label{alg:LSLRKF}
    \begin{algorithmic}[1]%\noindent 
		\State Initialize $\Ncal(\hat\xvec_{0}|\mxhat_{0}, \Pxxhat_{0}),\Pcal_{\F},\T_{\F},\Pcal_{\HH},\T_{\HH}$
		
		\For{$k=1$ {\bf to} $K$}
		
		\textit{\textbf{// Time update}}

		\State $\texttt{\{}\mxbar_{k-1}, \Pxxbar_{k-1}\texttt{\}} = \texttt{permute}(\T_{\F}, \mxhat_{k-1}, \Pxxhat_{k-1})$

		\State $\texttt{\{}\Lnnbar_{k-1}, \Lnlbar_{k-1}\texttt{\}} = \texttt{chol\_crout}(\Pxxbar_{k-1})$

        \State Compute ${\Zcal}_{k-1}$ by~\eqref{eq:subsetdef} using $\mnbar_{k-1},\Lnnbar_{k-1},\mathcal{P}_{\F}$

        \State Evaluate $\mxbar_{k|k-1}$ and $\Pxxbar_{k|k-1}$by Proposition~\ref{thm:main} with~\eqref{eq:dyn}

		\State $\texttt{\{}\mxhat_{k|k-1}, \Pxxhat_{k|k-1}\texttt{\}} \hspace{-2pt}= \hspace{-2pt}\texttt{permute}(\T_{\F}^{\top}, \mxbar_{k|k-1}, \Pxxbar_{k|k-1})$
        
        \textit{\textbf{// Measurement update}}
		\State $\texttt{[}\mxbar_{k|k-1}, \Pxxbar_{k|k-1}\texttt{]} = \texttt{permute}(\T_{\HH}, \mxhat_{k|k-1}, \Pxxhat_{k-1})$

		\State $\texttt{\{}\Lnnbar_{k|k-1}, \Lnlbar_{k|k-1}\texttt{\}} = \texttt{chol\_crout}(\Pxxbar_{k|k-1})$

        \State Compute ${\Zcal}_{k|k-1}$ by~\eqref{eq:subsetdef} using $\mnbar_{k|k-1},\Lnnbar_{k|k-1},\mathcal{P}_{\HH}$

        \State \parbox[t]{\dimexpr .94\linewidth-\algorithmicindent}{Evaluate $\{\mybar_{k|k-1}, \Pyybar_{k|k-1}, \Pxybar_{k|k-1}\}$ using Proposition~\ref{thm:main} for the function in~\eqref{eq:meas} with the point set $\mathcal{P}_{\HH}$}

        \State Evaluate $\{\mxbar_{k}, \Pxxbar_{k}\}$ using~\eqref{eq:conditional}

		\State $\texttt{\{}\mxhat_{k}, \Pxxhat_{k}\texttt{\}} = \texttt{permute}(\T_{\HH}^\top, \mxbar_{k}, \Pxxbar_{k})$
        \EndFor
	\end{algorithmic}
\end{algorithm}

\section{Numerical results}\label{sec:numerical}
In this section, we study the effects of using the simplified moment approximations in two primary respects. We first show that evaluating the original integral approximations in~\eqref{eq:approxmomentint} using the integration point sets defined in Section~\ref{sec:results} yield the same results, down to numerical precision, as when evaluating the same moment integrals using Proposition~\ref{thm:main}. In this first example, we also illustrate the the average computational times and relate these to the sizes of the dimensions of he linear and nonlinear states respectively, as well as the cardinality of the point sets. In the second example, we consider the SC integration point set defined in Section~\ref{sec:CKF}, and compare the third order CKF and the corresponding PL-CKF, showing two algorithms produce identical estimates down to machine precision.

\subsection{Complexity analysis and verification of Proposition 1}
Clearly, the permutations $\T$ used in the PL-LRKF only imply a change in indexation, and little if any extra computational cost. Consequently, the main difference in the two algorithms follows from Corollary~\ref{cor:compgains} and~\ref{cor:A}. Executing the PL-LRKF vs the LRKF should imply a significant reduction in the computational complexity if $Z_F=\text{dim}(\bar{\zvec}_{k}) \leq X$ and $Z_H=\text{dim}(\bar{\zvec}_{k|k-1}) \leq X$. Here, depending on the filter used and its associate scaling $C(X)$, the number of function evaluations required to execute the filter will be reduced by a factor $(C(Z_F) + C(Z_H))(2C(X))^{-1}$.  This potentially yields an extreme reduction in complexity for filters like the GHKF where $C(X)$ is exponential. In addition, this will only require evaluation of the nonlinear part of the flow and measurement function, and not repeated evaluations of the full functions $\F$ and $\HH$, further reducing the numerical complexity. Finally, if the conditions in Corollary~\ref{cor:A} are met, significant computational gains will be made as there will be no need for evaluating the full Cholesky decomposition.

To illustrate the implications of exploiting the linear substructure in the moment matching, we form a nonlinear function $\G(\xvec)$ defined as in~\eqref{eq:func}, here with random dense matrix $\A\in\mathbb{R}^{L\times X}$ and let $\gvec(\zvec) = \zvec + \|\zvec\|_2^2\mathbf{1}_{Z}\in\mathbb{R}^Z$. We execute the moment approximations using:
\begin{enumerate}
    \item[(A)] The original cubature rules defined in Section \ref{sec:results}
    \item[(B)] The corresponding cubature rules in Proposition~\ref{thm:main}, here denoted PL-SC, PL-UT and PL-GHC.
\end{enumerate}
A set of $10^6$ moment approximations are done for various pairs of $(Z,L)$. In order to check the correctness of Proposition~\ref{thm:main}, we compare the moments by a metric $\Delta(\my) = \mathbb{E}[\|\boldsymbol{m}^{\boldsymbol{y},A} - \boldsymbol{m}^{\boldsymbol{y},B}\|_2]$, where $\boldsymbol{m}^{\boldsymbol{y},A}$ and $\boldsymbol{m}^{\boldsymbol{y},B}$ denotes a moment, here the output mean, as computed by (A) and (B), respectively. In addition, we let ${t}_A$ and ${t}_B$ denote the mean computational time with (A) and (B) in Maltab running on an Intel i5-6200U CPU at 2.30GHz. The result is shown in Tables~\ref{tab:tab_CKF},~\ref{tab:tab_UKF}, and~\ref{tab:tab_GHKF}. Note that many of the fields in Table~\ref{tab:tab_GHKF} could not be filled, as Matlab ran out of memory when storing all of the associated integration points. In all of the tested moment approximations, $\Delta(\my), \Delta(\Pxy), \Delta(\Pyy)$ were all zero down to numerical precision. This shows the significant computational gains that can be made in using Proposition~\ref{thm:main} to exploit knowledge of linear substructure, and also indicates that the GHKF, which is commonly dismissed for high-dimensional state estimation, can be used provided the number of non-linear states $Z$ is relatively small.

\replaceinpaper{
\begin{table}[h]
    \centering
    \caption{Using (A) the original SC and (B) the PL-SC in $10^6$ moment approximations. Mean computational time [s], relative speed, and difference between the computed moments in the two-norm, for various dimensions $(Z,L)$.}
    \begin{tabular}{|c|c|c|c|c|c|c|}\hline
        $(Z/L)$ & $t_A$  & $t_B$ & $t_A/t_B$ & $\Delta(\my)$ & $\Delta(\Pxy)$&$\Delta(\Pyy)$\\\hline\hline
        (3/10) & $1.54\cdot{10}^{-4}$ & $7.64\cdot{10}^{-5}$ & \textbf{2.01} & $3.33\cdot{10}^{-15}$ & $9.67\cdot{10}^{-16}$ & $1.77\cdot{10}^{-14}$ \\
        %%%%%%%%%%%%%%%%%%%%%%%%%%%%%%%%%%%%%%%%%%%%%%%%%%%%%%%%%%%%%%%%%
        (3/100) & $4.21\cdot{10}^{-3}$ & $4.51\cdot{10}^{-4}$ & \textbf{9.34} & $2.85\cdot{10}^{-14}$ & $3.24\cdot{10}^{-15}$ & $3.77\cdot{10}^{-13}$ \\
        %%%%%%%%%%%%%%%%%%%%%%%%%%%%%%%%%%%%%%%%%%%%%%%%%%%%%%%%%%%%%%%%%
        (3/1000) & $1.50$ & $1.30\cdot10^{-1}$ & \textbf{11.66} & $1.87\cdot{10}^{-12}$ & $1.21\cdot{10}^{-14}$ & $8.41\cdot{10}^{-11}$ \\
        %%%%%%%%%%%%%%%%%%%%%%%%%%%%%%%%%%%%%%%%%%%%%%%%%%%%%%%%%%%%%%%%%
        (50/100) & $9.43\cdot{10}^{-3}$ & $3.23\cdot{10}^{-3}$ & \textbf{2.91} & $1.37\cdot{10}^{-12}$ & $2.73\cdot{10}^{-14}$ & $9.47\cdot{10}^{-11}$ \\
        (50/1000) &  1.69 & $1.15\cdot10^{-1}$ & \textbf{11.51} & $1.53\cdot{10}^{-11}$ & $9.11\cdot{10}^{-14}$ & $2.10\cdot{10}^{-10}$
        \\\hline
    \end{tabular}
    \label{tab:tab_CKF}
    \vspace{10pt}
    
    \caption{Using (A) the original UT and (B) the PL-UT in $10^6$ moment approximations. Mean computational time [s], relative speed, and difference between the computed moments in the two-norm, for various dimensions $(Z,L)$.}
    \begin{tabular}{|c|c|c|c|c|c|c|}\hline
        $(Z/L)$ & $t_A$  & $t_B$ & $t_A/t_B$ & $\Delta(\my)$ & $\Delta(\Pxy)$&$\Delta(\Pyy)$\\\hline\hline
        (3/10) & $1.51\cdot{10}^{-4}$ & $7.88\cdot{10}^{-5}$ & \textbf{2.00} & $6.24\cdot{10}^{-15}$ & $7.88\cdot{10}^{-16}$ & $3.42\cdot{10}^{-14}$ \\
        %%%%%%%%%%%%%%%%%%%%%%%%%%%%%%%%%%%%%%%%%%%%%%%%%
        (3/100) & $4.15\cdot{10}^{-3}$ & $4.23\cdot{10}^{-4}$ & \textbf{9.80} & $1.65\cdot{10}^{-14}$ & $3.19\cdot{10}^{-15}$ & $2.55\cdot{10}^{-12}$ \\
        %%%%%%%%%%%%%%%%%%%%%%%%%%%%%%%%%%%%%%%%%%%%%%%%%
        (3/1000) & $1.57$ & $1.20\cdot{10}^{-1}$ & \textbf{13.06} & $1.87\cdot{10}^{-13}$ & $1.19\cdot{10}^{-14}$ & $8.65\cdot{10}^{-11}$ \\
        %%%%%%%%%%%%%%%%%%%%%%%%%%%%%%%%%%%%%%%%%%%%%%%%%
        (50/100) & $9.86\cdot{10}^{-3}$ & $2.29\cdot{10}^{-3}$ & \textbf{4.30} & $1.34\cdot{10}^{-12}$ & $2.85\cdot{10}^{-14}$ & $8.65\cdot{10}^{-11}$ \\
        %%%%%%%%%%%%%%%%%%%%%%%%%%%%%%%%%%%%%%%%%%%%%%%%%
        (50/1000) & 1.77 &  $1.39\cdot{10}^{-1}$ & \textbf{12.78} &
        $1.33\cdot{10}^{-11}$ & $0.49\cdot{10}^{-14}$ & $2.99\cdot{10}^{-10}$ \\\hline
    \end{tabular}
    \label{tab:tab_UKF}
    \vspace{10pt}
    
    \caption{Using (A) the original GHC and (B) the PL-GHC in $10^6$ moment approximations. Mean computational time [s], relative speed, and difference between the computed moments in the two-norm, for various dimensions $(Z,L)$.}
    \begin{tabular}{|c|c|c|c|c|c|c|}\hline
        $(Z/L)$ & $t_A$  & $t_B$ & $t_A/t_B$ & $\Delta(\my)$ & $\Delta(\Pxy)$&$\Delta(\Pyy)$\\\hline\hline
        (3/3) & $5.11\cdot{10}^{-3}$ & $1.24\cdot{10}^{-4}$  & $\mathbf{4.12\cdot 10^1}$ &  $1.77\cdot{10}^{-14}$ &  $1.25\cdot{10}^{-15}$ &  $2.06\cdot{10}^{-14}$ \\
        %%%%%%%%%%%%%%%%%%%%%%%%%%%%%%%
        (3/4) & $1.82\cdot 10^{-2}$ & $1.12 \cdot 10^{-4}$  & $\mathbf{1.62\cdot 10^2}$ &
        $2.27\cdot{10}^{-14}$ &  $1.30\cdot{10}^{-15}$ &  $2.75\cdot{10}^{-14}$ \\
        %%%%%%%%%%%%%%%%%%%%%%%%%%%%%%%
        (3/5) & $2.12\cdot 10^{-1}$ & $1.23 \cdot 10^{-4}$ &  $\mathbf{1.73\cdot 10^3}$ &
        $4.41\cdot{10}^{-14}$ &  $2.05\cdot{10}^{-15}$ &  $3.09\cdot{10}^{-14}$ \\
        %%%%%%%%%%%%%%%%%%%%%%%%%%%%%%%
        (3/10) & - & $\mathbf{1.70 \cdot 10^{-4}}$ & - & - & - & - \\
        (3/100) & - & $\mathbf{2.50 \cdot 10^{-3}}$ & - & - & - & - \\\hline
    \end{tabular}
    \label{tab:tab_GHKF}
\end{table}
}{
\begin{table}[h]
    \centering
    \caption{Mean computational time [s] and relative speed of (A) CKF and (B) the PL-CKF in $10^6$ moment approximations.}
    \begin{tabular}{|c|c|c|c|}\hline
        $(Z/L)$ & $t_A$  & $t_B$ & $t_A/t_B$\\\hline\hline
        (3/10) & $1.54\cdot{10}^{-4}$ & $7.64\cdot{10}^{-5}$ & \textbf{2.01}  \\
        %%%%%%%%%%%%%%%%%%%%%%%%%%%%%%%%%%%%%%%%%%%%%%%%%%%%%%%%%%%%%%%%%
        (3/100) & $4.21\cdot{10}^{-3}$ & $4.51\cdot{10}^{-4}$ & \textbf{9.34} \\
        %%%%%%%%%%%%%%%%%%%%%%%%%%%%%%%%%%%%%%%%%%%%%%%%%%%%%%%%%%%%%%%%%
        (3/1000) & $1.50$ & $1.30\cdot10^{-1}$ & \textbf{11.66} \\
        %%%%%%%%%%%%%%%%%%%%%%%%%%%%%%%%%%%%%%%%%%%%%%%%%%%%%%%%%%%%%%%%%
        (50/100) & $9.43\cdot{10}^{-3}$ & $3.23\cdot{10}^{-3}$ & \textbf{2.91}  \\
        (50/1000) &  1.69 & $1.15\cdot10^{-1}$ & \textbf{11.51}
        \\\hline
    \end{tabular}
    \label{tab:tab_CKF}
    \centering
    \caption{Mean computational time [s] and relative speed of (A) UKF and (B) the PL-UKF in $10^6$ moment approximations.}
    \begin{tabular}{|c|c|c|c|}\hline
        $(Z/L)$ & $t_A$  & $t_B$ & $t_A/t_B$ \\\hline\hline
        (3/10) & $1.51\cdot{10}^{-4}$ & $7.88\cdot{10}^{-5}$ & \textbf{2.00} \\
        %%%%%%%%%%%%%%%%%%%%%%%%%%%%%%%%%%%%%%%%%%%%%%%%%
        (3/100) & $4.15\cdot{10}^{-3}$ & $4.23\cdot{10}^{-4}$ & \textbf{9.80} \\
        %%%%%%%%%%%%%%%%%%%%%%%%%%%%%%%%%%%%%%%%%%%%%%%%%
        (3/1000) & $1.57$ & $1.20\cdot{10}^{-1}$ & \textbf{13.06}  \\
        %%%%%%%%%%%%%%%%%%%%%%%%%%%%%%%%%%%%%%%%%%%%%%%%%
        (50/100) & $9.86\cdot{10}^{-3}$ & $2.29\cdot{10}^{-3}$ & \textbf{4.30}  \\
        %%%%%%%%%%%%%%%%%%%%%%%%%%%%%%%%%%%%%%%%%%%%%%%%%
        (50/1000) & 1.77 &  $1.39\cdot{10}^{-1}$ & \textbf{12.78} \\\hline
    \end{tabular}
    \label{tab:tab_UKF}
    \centering
    \caption{Mean computational time [s] and relative speed of (A) GHCR and (B) the PL-GHCR in $10^6$ moment approximations.}
    \begin{tabular}{|c|c|c|c|}\hline
        $(Z/L)$ & $t_A$  & $t_B$ & $t_A/t_B$\\\hline\hline
        (3/3) & $5.11\cdot{10}^{-3}$ & $1.24\cdot{10}^{-4}$  & $\mathbf{4.12\cdot 10^1}$\\
        %%%%%%%%%%%%%%%%%%%%%%%%%%%%%%%
        (3/4) & $1.82\cdot 10^{-2}$ & $1.12 \cdot 10^{-4}$  & $\mathbf{1.62\cdot 10^2}$\\
        %%%%%%%%%%%%%%%%%%%%%%%%%%%%%%%
        (3/5) & $2.12\cdot 10^{-1}$ & $1.23 \cdot 10^{-4}$ &  $\mathbf{1.73\cdot 10^3}$\\
        %%%%%%%%%%%%%%%%%%%%%%%%%%%%%%%
        (3/10) & - & $\mathbf{1.70 \cdot 10^{-4}}$ & - \\
        (3/100) & - & $\mathbf{2.50 \cdot 10^{-3}}$ & -\\\hline
    \end{tabular}
    \label{tab:tab_GHKF}
\end{table}
}

\subsection{Application example}
Next, we give an example where Algorithm~\ref{alg:LSLRKF} can be of particular use. Imagine a scenario where $N$ agents with positions $\pvec_k^i\in \mathbb{R}_k^3, \;i = 1,...,N$, velocities $\vvec^i_k\in \mathbb{R}^3$ and accelerations $\avec_k^i\in \mathbb{R}^3$ in the vicinity of a base station located at $\pvec^B = \Z\in\mathbb{R}^3$. The  agents could here be ground or aerial vehicles, and for the sake of generality, suppose their motions are described by a discrete time Singer model, defined as
\begin{equation}
\xvec_{k+1}^i = \A_k^i\xvec_k^i + \qvec_k^i, \qquad \qvec_k^i\sim\Ncal(\Z,\Q_k^i),
\end{equation}
where $\A_k^i$ and $\Q_k^i$ are linear maps, stated explicitly in~\cite{singer1970estimating}.

Each of the agents estimate its own states independently based on local information. This could be done using a wide variety of algorithms, and we assume this local estimate can be represented in its first two moments, as $\Ncal(\xvec_k^i|,{\mxhat_k}^i, {\Pxxhat_k}^i)$, and that this information is occasionally transmitted to a base station. The base station subsequently fuses this information together with two-dimensional bearing angle measurements taken of each agent. In the base station, the state vector therefore takes the form $\xvec_k = [(\xvec_k^1)^{\top},...,(\xvec_k^N)^{\top}]^{\top}$, with the estimation model
\begin{equation}\label{eq:exampledyn}
\xvec_{k+1} = \A_k\xvec_k + \qvec_k, \qquad \qvec_k\sim\Ncal(\Z,\Q_k),
\end{equation}
where $\A_k = \I_N\otimes \A_k^i$ and $\Q_k = \I_N\otimes\Q_k^i$. The angular measurement model is defined as $\yvec_{\alpha} = \boldsymbol{\alpha}(\pvec_k^i)+ \rvec_{Z,k}^i$, with
\begin{equation}
\boldsymbol{\alpha}(\pvec_k^i) = \begin{bmatrix}
\arctan\Big(\frac{(\pvec_k^i)^{\top}\evec_2^3 }{(\pvec_k^i)^{\top}\evec_1^3}\Big)\\
\arctan\Big(\frac{\sqrt{((\pvec_k^i)^{\top}\evec_1^3) + ((\pvec_k^i)^{\top}\evec_2^3)^2}}{ (\pvec_k^i)^{\top}\evec_3^3}\Big)
\end{bmatrix},
\end{equation}
and the noise $\rvec_{Z,k}^i\sim \mathcal{N}(\Z, \sigma_{\alpha}^2\I_2)$. Thus, by letting
\begin{equation*}
\pvec_k \hspace{-3pt}=\hspace{-3pt} \begin{bmatrix}
\pvec_k^1\\\vdots\\\pvec_k^N
\end{bmatrix}\hspace{-2pt}, \;
\hvec(\pvec_k) \hspace{-3pt}=\hspace{-3pt}
\begin{bmatrix}
\boldsymbol{\alpha}(\pvec_k^1)\\
\vdots\\
\boldsymbol{\alpha}(\pvec_k^N)\\
\end{bmatrix}\hspace{-2pt}, \;\R_{x,k} \hspace{-3pt}=\hspace{-3pt} \begin{bmatrix}
{\Pxxhat_k}^1 & \hspace*{-5pt}\cdots\hspace*{-5pt} & \Z\\
\vdots & \hspace*{-5pt}\ddots\hspace*{-5pt} & \vdots \\
\Z & \hspace*{-5pt}\cdots\hspace*{-5pt} & {\Pxxhat_k}^N
\end{bmatrix}
\end{equation*}
and $\R_{\alpha,k} = \sigma_\alpha^2\I_{2N}$, the combined measurement model of all agents' relative angles can be written
\begin{align}\label{eq:examplemeasurement}
&\yvec_k = \HH(\xvec_k) + \rvec_k
=
\begin{bmatrix}
\hvec(\pvec_k)\\
\xvec_k
\end{bmatrix} + 
\begin{bmatrix}
\rvec_{\alpha,k}\\
\rvec_{x,k}
\end{bmatrix},
\\&\Ncal(\rvec_k|\Z,\R_k) = \Ncal\Bigg(
\begin{bmatrix}
\rvec_{\alpha,k}\\
\rvec_{x,k}
\end{bmatrix}
\Bigg|
\begin{bmatrix}\Z\\\Z\end{bmatrix}
,
\begin{bmatrix}
\R_{\alpha,k}&\Z\\
\Z&\R_{x,k}
\end{bmatrix}
\Bigg),\nonumber
\end{align}
where the transformation
\begin{equation}\label{eq:transformation}
\T_{\HH} = \begin{bmatrix}
\I_N\otimes\begin{bmatrix}
\I_3& \Z_{3\times 6}\\
\end{bmatrix}\\
\I_N\otimes\begin{bmatrix}
\Z_{6\times 3}&\I_6
\end{bmatrix}
\end{bmatrix},
\end{equation}
transforms the state $\xvec_k$ into the form in~\eqref{eq:distX}, where the nonlinear states (here the positions) are stacked in the first $3N$ elements of the transformed state vector $\bar{\xvec}_k = \T_{\HH}\xvec_k$. With the dynamics in~\eqref{eq:exampledyn} and the partially linear measurement equation in~\eqref{eq:examplemeasurement}, we consider a problem where $N = 10$ (that is, $Z = 30$ and $L = 60$), and execute
\begin{enumerate}
    \item[(A)] a CKF, that is, Algorithm~\ref{alg:LRKF} with point sets $\Pcal_F$ and $\Pcal_H$ defined according to the SC in Section~\ref{sec:CKF}.
    \item[(B)] a PL-CKF, that is, Algorithm~\ref{alg:LSLRKF} with point sets $\Pcal_F$ and $\Pcal_H$ defined according to the SC in Section~\ref{sec:CKF}, using $\T_{\F} = \I$ and $\T_{\HH}$ defined in~\eqref{eq:transformation}, and evaluating the moment integrals using Proposition~\ref{thm:main}.
\end{enumerate}

The resulting mean estimate error of the PL-CKF is shown in terms of the positions, velocities and accelerations of all agents together with the $95\%$-confidence interval in Figure~\ref{fig:esterr}. Furthermore, as we compute the same joint distribution in the PL-CKF and the CKF down to numerical precision, we get the same estimates with the two algorithms. This is shown by comparing the posterior estimate means, $\boldsymbol{m}^{\boldsymbol{x},A}_{k|k}$ for the CKF and $\boldsymbol{m}^{\boldsymbol{x},B}_{k|k}$ for the PL-CKF, here in $l_2$-norm in time, as depicted in Figure~\ref{fig:momerr}. This demonstrates that Algorithm~\ref{alg:LSLRKF} indeed works as intended, and offers further numerical verification of the main result in Proposition~\ref{thm:main}.

\begin{figure}[h!]
    \centering
    \includegraphics[width=0.8\textwidth]{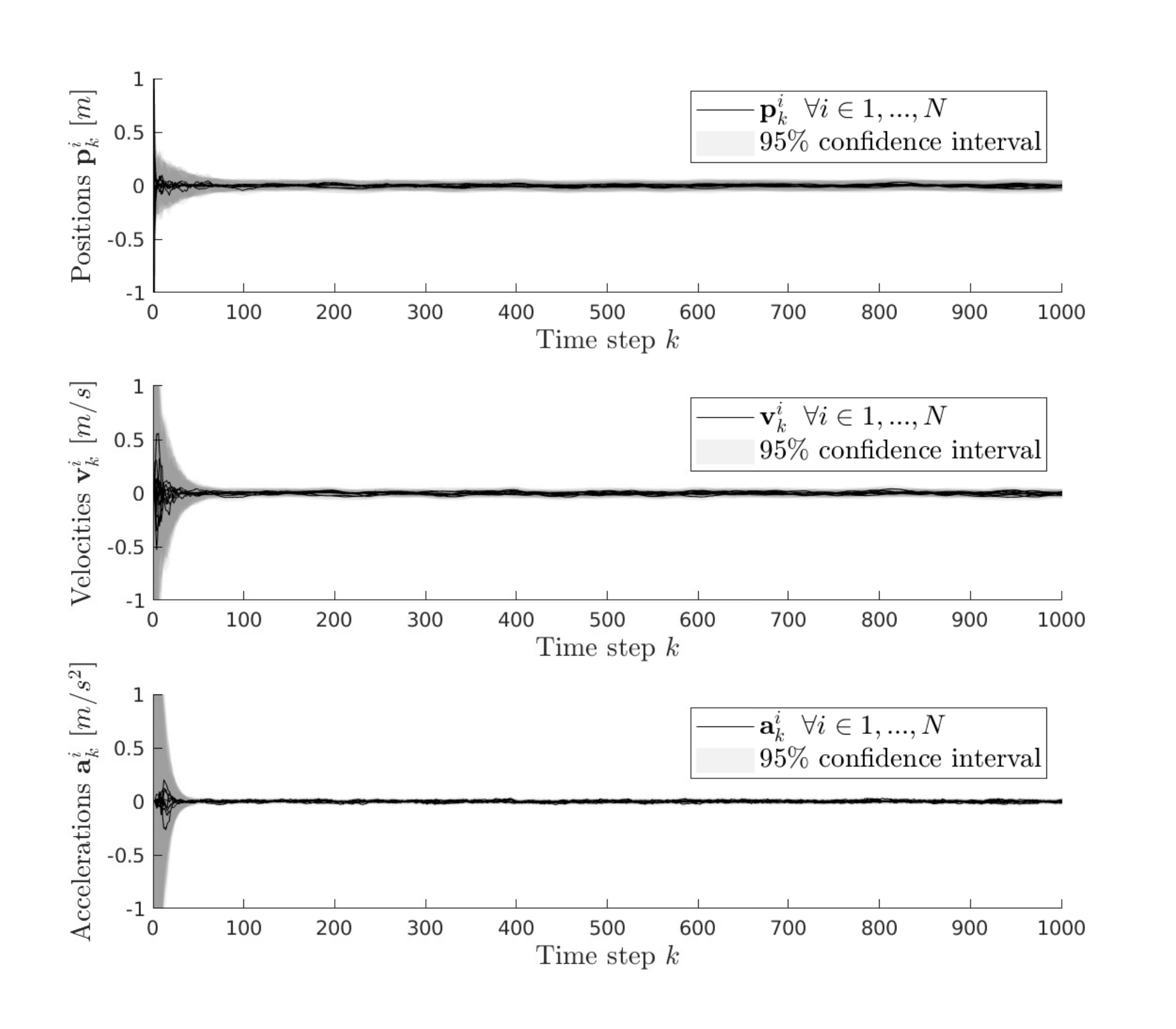}\vspace*{-20pt}
    \caption{Mean estimate error in positions (top), velocities (center) and accelerations (bottom), when running the PL-CKF for $N=10$ agents.}
    \label{fig:esterr}
    \vspace*{5pt}
    \includegraphics[width=0.75\textwidth]{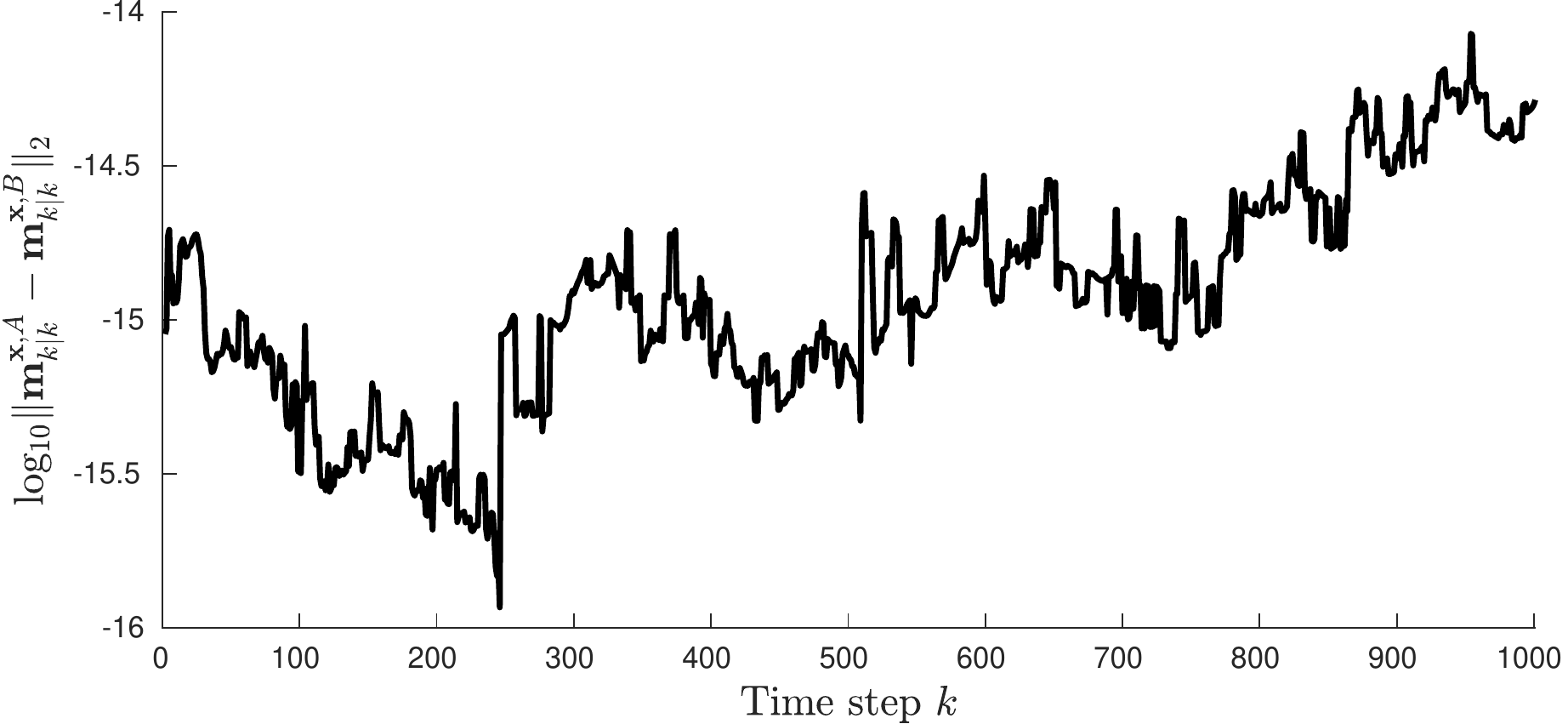}\vspace*{-10pt}
    \caption{Difference in the first moment of the state estimate distributions in the 10-logarithm and $l_2$-norm when running the CKF and the PL-CKF on the same synthetic data with $N = 10$ agents.}
    \label{fig:momerr}
\end{figure}

\section{Conclusions}\label{sec:conclusion}
In this paper, we examine the implications of known linear substructure in the LRKFs, which can be leveraged to decrease the computational complexity of the filters significantly. We have stated exactly how the evaluation of the moment matching approximations in~\eqref{eq:approxmomentint} simplifies when there exists a known linear substructure in Proposition~\ref{thm:main}, provided the cubature rule satisfies Assumptions~\ref{rem:symmetry}-\ref{rem:product}. In addition, we have shown that we need not use the entire set of points $\Pcal_z$ in evaluating these integrals provided additional symmetry assumption, and that the number of function evaluations then scales as $C(Z)$ instead of $C(X)$. We have given a condition under which only a partial Cholesky decomposition needs to be computed in the moment matching. Furthermore, we have demonstrated that the above assumptions are satisfied for the cubature rules used in the CKF, UKF, and GHKF. In doing so, we have also given modified versions of the original algorithms which exploit this linear substructure in Algorithm~\ref{alg:LSLRKF}, and demonstrated its efficacy on a distributed filtering example. The results can easily be generalized to smoothing problems, but this falls outside the scope of the paper and is left as a topic for future research. On a final note, known linear substructure should be exploited in filtering applications whenever possible, and any implementation of an LRKF should be done using the moment approximations in Proposition~\ref{thm:main} in order to conserve computational resources.

\clearpage\section{Appendix}\label{sec:app}

\begin{proof}[Proof of Proposition~\ref{thm:main}]

Granted Assumption~\ref{rem:symmetry}, the cardinality of the sets of linear and nonlinear points are necessarily even, and therefore, we can \removeinpaper{without loss of generality} choose an ordering such that
\begin{align}
\Xi_l^{(i)} + \Xi_l^{(|\Pcal_l|/2+i)}=\Z\quad\forall i &= \{1,...,|\Pcal_l|/2\},&
\Xi_z^{(i)} + \Xi_z^{(|\Pcal_z|/2+i)}=\Z\quad\forall i &= \{1,...,|\Pcal_z|/2\}.
\end{align}

Now, consider the part of the point $\Xcal^{(i)}$ that is the input to the nonlinear function $\fullnonlinfunc$, as defined in~\eqref{eq:func}. We have that
\begin{equation}
\Zcal^{(i)} \hspace{-2pt}=\hspace{-2pt} \mathbf{N}\Xcal^{(i)}\removeinpaper{= \mathbf{N}\begin{bmatrix}
\mn\\
\ml
\end{bmatrix} + \mathbf{N}
\begin{bmatrix}
\Lnn&\Z\\
\Lln&\Lll
\end{bmatrix}
\xivec^{(i)}}.
%\hspace{-2pt}=\hspace{-2pt}\begin{cases}
%\mn + \Lnn\mathbf{N}\xivec&\text{if}\;(w,\xivec)\hspace{-2pt}\in\hspace{-2pt}\Pcal_z\\
%\mn &\text{if}\;(w,\xivec)\hspace{-2pt}\in\hspace{-2pt}\Pcal_c\cup\Pcal_l
%\end{cases}\hspace{-3pt}.
\end{equation}
Then, since $\mathbf{N}\Zcal_c =\mathbf{N}\Zcal_l =  \mvec^{\zvec}$, we obtain
\shortenedEquation{
\begin{subequations}\label{eq:Z}
\begin{align}
\sum\limits_{i = 1}^{|\Pcal|}
w^{(i)}\nonlinfunc(\Zcal^{(i)})
&=
\sum\limits_{i = 1}^{|\Pcal_c|} w^{(i)}_c\nonlinfunc(\Zcal_c^{(i)}) +
\sum\limits_{i = 1}^{|\Pcal_z|} w^{(i)}_z\nonlinfunc(\Zcal_z^{(i)}) + \sum\limits_{i = 1}^{|\Pcal_l|}
w^{(i)}_l\nonlinfunc(\Zcal_l^{(i)})\\
&=\Big(\sum\limits_{i = 1}^{|\Pcal_c|}
w^{(i)}_c+
\sum\limits_{i = 1}^{|\Pcal_l|}
w^{(i)}_l\Big)\nonlinfunc(\mn)+
\sum\limits_{i = 1}^{|\Pcal_z|}
w^{(i)}_z\nonlinfunc(\Zcal_z^{(i)})\\
&=
w_{cl}\nonlinfunc(\mn)+
\sum\limits_{i = 1}^{|\Pcal_z|}
w^{(i)}_z\nonlinfunc(\Zcal_z^{(i)}).
\end{align}
\end{subequations}
}{
\begin{align}\label{eq:Z}
\sum\limits_{i = 1}^{|\Pcal|}
w^{(i)}\nonlinfunc(\Zcal^{(i)})&=
w_{cl}\nonlinfunc(\mn)+
\sum\limits_{i = 1}^{|\Pcal_z|}
w^{(i)}_z\nonlinfunc(\Zcal_z^{(i)}).
\end{align}
}
where, by Assumption~\ref{rem:productFirst},
\begin{equation}
w_{cl} =\hspace{-10pt}
\sum_{(w,\xivec)\in(\Pcal_c\cup \Pcal_l)}\hspace{-5pt}
w = 1-\hspace{-5pt}\sum_{(w,\xivec)\in(\Pcal_z)}\hspace{-5pt}
w = 1-\wvec_z \mathbf{1}_{|\Pcal_z|}.
\end{equation}
In addition, using Assumptions~\ref{rem:symmetry} and~\ref{rem:productFirst}, we have that
\shortenedEquation{
\begin{subequations}\label{eq:X}
\begin{align}
\sum\limits_{i = 1}^{|\Pcal|}w^{(i)}\mathcal{X}^{(i)} &= 
\sum\limits_{i = 1}^{|\Pcal_c|}w_c^{(i)}\mathcal{X}_c^{(i)}+
\sum\limits_{i = 1}^{|\Pcal_z|}w_z^{(i)}\mathcal{X}_z^{(i)}+
\sum\limits_{i = 1}^{|\Pcal_l|}w_l^{(i)}\mathcal{X}_l^{(i)}\\
&=
\sum\limits_{i = 1}^{|\Pcal_c|}w_c^{(i)}\mx+
\sum\limits_{i = 1}^{|\Pcal_z|}w_z^{(i)}(\mx + \Lxx\Xi_z^{(i)})+
\sum\limits_{i = 1}^{|\Pcal_l|}w_l^{(i)}(\mx + \Lxx\Xi_l^{(i)})\\
&=
\sum\limits_{i = 1}^{|\Pcal_c|}w_c^{(i)}\mx+
\sum\limits_{i = 1}^{|\Pcal_z|/2}w_z^{(i)}(\mx + \Lxx\Xi_z^{(i)} + \mx - \Lxx\Xi_z^{(i)})\\&\hspace{67pt}
+\sum\limits_{i = 1}^{|\Pcal_l|/2}w_l^{(i)}(\mx + \Lxx\Xi_l^{(i)} + \mx - \Lxx\Xi_l^{(i)})\\
&=
\Big(\sum\limits_{i = 1}^{|\Pcal_c|}w_c^{(i)}+
\sum\limits_{i = 1}^{|\Pcal_z|/2}2w_z^{(i)}+
\sum\limits_{i = 1}^{|\Pcal_l|/2}2w_l^{(i)}\Big)\mx\\
&=
\Big(\sum\limits_{i = 1}^{|\Pcal_c|}w_c^{(i)}+
\sum\limits_{i = 1}^{|\Pcal_z|}w_z^{(i)}+
\sum\limits_{i = 1}^{|\Pcal_l|}w_l^{(i)}\Big)\mx\\
&=
\Big(\sum\limits_{i = 1}^{|\Pcal|}w^{(i)}\Big)\mx\\
&=\mx.
\end{align}
\end{subequations}
}{
\begin{equation}\label{eq:X}
    \sum\limits_{i = 1}^{|\Pcal|}w^{(i)}\mathcal{X}^{(i)} =\mx.
\end{equation}
}
Using~\eqref{eq:Z} and~\eqref{eq:X}, we can express the output mean as
\shortenedEquation{
\begin{subequations}\label{eq:meanY}
\begin{align}
\my &=\int_{\mathbb{R}^X}\fullnonlinfunc(\xvec)\mathcal{N}(\xvec|\mx,\Pxx)d\xvec\\
&\approx\sum\limits_{i = 1}^{|\Pcal|}w^{(i)}\Ycal^{(i)}\\
%&=\sum\limits_{i = 1}^{|\Pcal|}w^{(i)}\fullnonlinfunc(\mathcal{X}^{(i)})\\
&=\sum\limits_{i = 1}^{|\Pcal|}w^{(i)}\begin{bmatrix}\nonlinfunc(\mathcal{Z}^{(i)})\\
\A\mathcal{X}^{(i)}\\
\end{bmatrix}\\
&=\begin{bmatrix}
\sum\limits_{i = 1}^{|\Pcal|}w^{(i)}\nonlinfunc(\mathcal{Z}^{(i)})\\
\A \sum\limits_{i = 1}^{|\Pcal|}w^{(i)}\mathcal{X}^{(i)}\\
\end{bmatrix}\\
&=\begin{bmatrix} w_{cl}\nonlinfunc(\mn)+
\sum\limits_{i = 1}^{|\Pcal_z|}
w^{(i)}_z\nonlinfunc(\Zcal_z^{(i)})\\
\A\mx
\end{bmatrix}\in\mathbb{R}^{Y}.
\end{align}
\end{subequations}
}{
\begin{align}\label{eq:meanY}
\my\hspace{-2pt}\approx\hspace{-2pt}\sum\limits_{i = 1}^{|\Pcal|}w^{(i)}\Ycal^{(i)}\hspace{-2pt}=\hspace{-2pt}\begin{bmatrix} w_{cl}\nonlinfunc(\mn)\hspace{-1pt}+\hspace{-1pt}
\sum\limits_{i = 1}^{|\Pcal_z|}
w^{(i)}_z\nonlinfunc(\Zcal_z^{(i)})\\
\A\mx
\end{bmatrix}\hspace{-2pt}.
\end{align}}
%%%%%%%%%%%%%%%%%%%%%%%%%%%%%%%%%%%%%%%%%%%%%%%%%%%%%%%%%%%%%%%%%%%%%%%%%%%%

Furthermore, by the definition of the point $\mathcal{X}^{(j)}$, its difference with the mean $\mx$ can be written
\begin{align}
\mathcal{X}^{(j)} \hspace{-1pt}-\hspace{-1pt} \mx &= \mx \hspace{-1pt}+\hspace{-1pt} \Lxx\xivec^{(j)} \hspace{-1pt}-\hspace{-1pt} \mx \hspace{-1pt}=\hspace{-1pt} \Lxx\xivec^{(j)}\hspace{-1pt}\in\hspace{-1pt}\mathbb{R}^X,
\end{align}
or, equivalently stated in matrix form,
\begin{equation}\label{eq:diffX}
\Xcalbf - \mx\mathbf{1}_{|\Pcal|}^{\top}=\removeinpaper{
\begin{bmatrix}
\mathcal{X}^{(1)} - \mx  & \dots & \mathcal{X}^{(|\Pcal|)} - \mx 
\end{bmatrix} = }
\Lxx\Xicalbf\in\mathbb{R}^{X\times |\Pcal|}.
\end{equation}

Similarly, using the expression of the output mean in~\eqref{eq:meanY}, and invoking the expression for the difference between input point and the input mean in~\eqref{eq:diffX}, the difference between an output point
$\mathcal{Y}^{(i)}$ and the output mean is
\shortenedEquation{
\begin{subequations}
\begin{align}
\mathcal{Y}^{(j)} - \my &\approx\fullnonlinfunc(\mathcal{X}^{(j)}) - \my
\\
&=\begin{bmatrix}
\nonlinfunc({\mathcal{Z}}^{(j)})\\
\A \mathcal{X}^{(j)}
\end{bmatrix}
-
\begin{bmatrix}
w_{cl}\nonlinfunc(\mn)+
\sum\limits_{i = 1}^{|\Pcal_z|}
w^{(i)}_z\nonlinfunc(\Zcal_z^{(i)})\\
\A\mx\
\end{bmatrix}\\
&=\begin{bmatrix}
\nonlinfunc({\mathcal{Z}}^{(j)})-w_{cl}\nonlinfunc(\mn)-
\sum\limits_{i = 1}^{|\Pcal_z|}
w^{(i)}_z\nonlinfunc(\Zcal_z^{(i)})\\
\A (\mathcal{X}^{(j)}-\mx)\
\end{bmatrix}\\
&=\begin{bmatrix}
\nonlinfunc({\mathcal{Z}}^{(j)})-w_{cl}\nonlinfunc(\mn)-
\sum\limits_{i = 1}^{|\Pcal_z|}
w^{(i)}_z\nonlinfunc(\Zcal_z^{(i)})\\
\A\Lxx\xivec^{(j)}
\end{bmatrix}
\in\mathbb{R}^{Y},
\end{align}
\end{subequations}
}{
\begin{align}
\mathcal{Y}^{(j)} \hspace{-1pt}-\hspace{-1pt}\my &\hspace{-2pt}\approx\hspace{-2pt}\begin{bmatrix}
\nonlinfunc({\mathcal{Z}}^{(j)})\hspace{-2pt}-\hspace{-2pt}w_{cl}\nonlinfunc(\mn)\hspace{-2pt}-\hspace{-2pt}
\sum\limits_{i = 1}^{|\Pcal_z|}
w^{(i)}_z\nonlinfunc(\Zcal_z^{(i)})\\
\A\Lxx\xivec^{(j)}
\end{bmatrix}\hspace{-4pt},\hspace{-4pt}
\end{align}
}
which may be written in matrix form as
\begin{equation}\label{eq:diffY}
\Ycalbf - \my\mathbf{1}_{|\Pcal|}^{\top}\removeinpaper{=
\begin{bmatrix}
\mathcal{Y}^{(1)} - \my  & \dots & \mathcal{Y}^{(|\Pcal|)} - \my 
\end{bmatrix} = 
\begin{bmatrix}
\Gcalbf+\Ccalbf\\
\A \Lxx \Xicalbf
\end{bmatrix}}\hspace{-2pt}=\hspace{-2pt}
\begin{bmatrix}
\Ccalbf_c&\Gcalbf_z+\Ccalbf_z&\Ccalbf_l\\
\A \Lxx \Xicalbf_c & \A \Lxx \Xicalbf_z & \A \Lxx \Xicalbf_l
\end{bmatrix}\hspace{-2pt}%\in\hspace{-2pt}\mathbb{R}^{Y\times |\Pcal|}
\end{equation}
\begin{subequations}\label{eq:diffY2}
where
\vspace*{-7pt}\begin{align}
\uvec &= -w_{cl}\nonlinfunc(\mn) - \sum\limits_{i = 1}^{|\Pcal_z|}
w^{(i)}_z\nonlinfunc(\Zcal_z^{(i)})\\
\Ccalbf_c &= (\nonlinfunc(\mn)+\uvec)\mathbf{1}_{|\Pcal_c|}^{\top}\\
\Ccalbf_l &= (\nonlinfunc(\mn)+\uvec)\mathbf{1}_{|\Pcal_l|}^{\top}\\
\Ccalbf_z &= \uvec\mathbf{1}_{|\Pcal_z|}^{\top}\\
\Gcalbf_z &= \begin{bmatrix}
\nonlinfunc(\Zcal^{(1)}) & \cdots & \nonlinfunc(\Zcal^{(|\Pcal_z|)})
\end{bmatrix}
\end{align}
\end{subequations}

Using~\eqref{eq:diffX} and~\eqref{eq:diffY}, we get the following compact expressions for the second moments of the joint distribution,
\shortenedEquation{
\begin{subequations}\label{eq:thmmainPxy}
\begin{align}
(\Pxy)^{\top}&=\int_{\mathbb{R}^X}(\fullnonlinfunc(\xvec) -\my)(\xvec -\mx)^{\top}\mathcal{N}(\xvec|\mx,\Pxx)d\xvec,\\
&\approx
\sum\limits_{i = 1}^{|\mathcal{P}|}w^{(i)}(\mathcal{Y}^{(i)} -\my)(\mathcal{X}^{(i)} -\mx)^{\top}\\
&=
\begin{bmatrix}
\mathcal{Y}^{(1)} - \my  & \dots & \mathcal{Y}^{(|\Pcal|)} - \my 
\end{bmatrix}\Wcalbf
\begin{bmatrix}
\mathcal{X}^{(1)} - \mx  & \dots & \mathcal{X}^{(|\Pcal|)} - \mx
\end{bmatrix}^{\top}\\
&=\begin{bmatrix}
(\Gcalbf+\Ccalbf)\\
\A \Lxx \Xicalbf
\end{bmatrix}
\Wcalbf\Xicalbf^{\top}{\Lxx}^{\top}=
\begin{bmatrix}
(\Gcalbf+\Ccalbf)\Wcalbf\Xicalbf^{\top}{\Lxx}^{\top}\\
\A \Lxx\Xicalbf\Wcalbf\Xicalbf^{\top}{\Lxx}^{\top}
\end{bmatrix}=\begin{bmatrix}
(\Gcalbf+\Ccalbf)\Wcalbf\Xicalbf^{\top}{\Lxx}^{\top}\\
\A \Pxx
\end{bmatrix}\\
\end{align}
\end{subequations}
}{
\begin{subequations}\label{eq:thmmainPxy}
\begin{align}
(\Pxy)^{\top}&\approx\removeinpaper{
\sum\limits_{i = 1}^{|\mathcal{P}|}w^{(i)}(\mathcal{Y}^{(i)} -\my)(\mathcal{X}^{(i)} -\mx)^{\top}=}\begin{bmatrix}
(\Gcalbf+\Ccalbf)\Wcalbf\Xicalbf^{\top}{\Lxx}^{\top}\\
\A \Pxx
\end{bmatrix}\\\label{eq:thmmainPyy}
\Pyy&\approx\removeinpaper{
\sum\limits_{i = 1}^{|\mathcal{P}|}w^{(i)}(\mathcal{Y}^{(i)} -\my)(\mathcal{Y}^{(i)} -\my)^{\top}
=}\begin{bmatrix}
(\Gcalbf+\Ccalbf)\Wcalbf(\Gcalbf+\Ccalbf)^{\top} &
\star\\
\A \Lxx \Xicalbf\Wcalbf(\Gcalbf+\Ccalbf)^{\top} & 
\A \Pxx\A^{\top}
\end{bmatrix}.
\end{align}
\end{subequations}
}
\removeinpaper{
and
\begin{subequations}\label{eq:thmmainPyy}
\begin{align}
\Pyy&\approx
\sum\limits_{i = 1}^{|\mathcal{P}|}w^{(i)}(\mathcal{Y}^{(i)} -\my)(\mathcal{Y}^{(i)} -\my)^{\top}\\
&=
(\mathcal{Y} - \my\mathbf{1}_{|\Pcal|}^{\top})\Wcalbf(\mathcal{Y} - \my\mathbf{1}_{|\Pcal|}^{\top}
)^{\top}\\
&=
\begin{bmatrix}
(\Gcalbf+\Ccalbf)\\
\A \Lxx \Xicalbf
\end{bmatrix}
\Wcalbf
\begin{bmatrix}
(\Gcalbf+\Ccalbf)\\
\A \Lxx \Xicalbf
\end{bmatrix}^{\top}\\
&=
\begin{bmatrix}
\Gcalbf+\Ccalbf\\
\A \Lxx \Xicalbf
\end{bmatrix}\Wcalbf
\begin{bmatrix}
(\Gcalbf+\Ccalbf)^{\top}\\
 \Xicalbf^{\top} {\Lxx}^{\top}\A^{\top}
\end{bmatrix}\\
&=
\begin{bmatrix}
(\Gcalbf+\Ccalbf)\Wcalbf(\Gcalbf+\Ccalbf)^{\top} &
(\Gcalbf+\Ccalbf)\Wcalbf\Xicalbf^{\top} {\Lxx}^{\top}\A^{\top}\\
\A \Lxx \Xicalbf\Wcalbf(\Gcalbf+\Ccalbf)^{\top} & 
\A \Pxx\A^{\top}
\end{bmatrix}
\end{align}
\end{subequations}
}
Now, consider the term $\Xicalbf\Wcalbf(\Gcalbf+\Ccalbf)^{\top}$. Using the previous definitions, we can express this term as
\shortenedEquation{
\begin{subequations}
\begin{align}
\Xicalbf\Wcalbf(\Gcalbf+\Ccalbf)^{\top}&=
\underbrace{\sum\limits_{i = 1}^{|\Pcal_c|}\Xi_c^{(i)}w_c^{(i)}(\nonlinfunc(\mn)+\uvec)^{\top}}_{=\Z\;\text{ as }\Xi_c^{(i)}=\Z\text{ by Definition}} + 
\sum\limits_{i = 1}^{|\Pcal_z|}\Xi_z^{(i)}w_z^{(i)}(\nonlinfunc(\Zcal^{(i)}) + \uvec)^{\top} + 
\sum\limits_{i = 1}^{|\Pcal_l|}\Xi_l^{(i)}w_l^{(i)}(\nonlinfunc(\mn)+\uvec)^{\top}\\
&=\sum\limits_{i = 1}^{|\Pcal_z|}\Xi_z^{(i)}w_z^{(i)}(\nonlinfunc(\Zcal^{(i)}) + \uvec)^{\top} + 
\underbrace{\sum\limits_{i = 1}^{|\Pcal_l|/2}(\Xi_l^{(i)} - \Xi_l^{(i)})w_l^{(i)}(\nonlinfunc(\mn)+\uvec)^{\top}}_{=\Z\text{ by Assumption~\ref{rem:symmetry}}}\\
&=\sum\limits_{i = 1}^{|\Pcal_z|/2}\Xi_z^{(i)}w_z^{(i)}(\nonlinfunc(\Zcal^{(i)}) + \uvec)^{\top} + \Xi_z^{(|\Pcal_z|/2+i)}w_z^{(|\Pcal_z|/2+i)}(\nonlinfunc(\Zcal^{(|\Pcal_z|/2+i)}) + \uvec)^{\top}\\
&=\sum\limits_{i = 1}^{|\Pcal_z|/2}\Xi_z^{(i)}w_z^{(i)}(\nonlinfunc(\Zcal^{(i)}) + \uvec)^{\top} - \Xi_z^{(i)}w_z^{(i)}(\nonlinfunc(\Zcal^{(|\Pcal_z|/2+i)}) - \uvec)^{\top}\\
&=\sum\limits_{i = 1}^{|\Pcal_z|/2}\Xi_z^{(i)}w_z^{(i)}(\nonlinfunc(\Zcal^{(i)}) - \uvec - \nonlinfunc(\Zcal^{(|\Pcal_z|/2+i)}) + \uvec)^{\top}\\
&=\sum\limits_{i = 1}^{|\Pcal_z|/2}\Xi_z^{(i)}w_z^{(i)}(\nonlinfunc(\Zcal^{(i)}) - \nonlinfunc(\Zcal^{(|\Pcal_z|/2+i)}))^{\top}\\
&=\sum\limits_{i = 1}^{|\Pcal_z|}\Xi_z^{(i)}w_z^{(i)}\nonlinfunc(\Zcal^{(i)})^{\top}\\
&=\Xicalbf_z\Wcalbf_z\Gcalbf_z^{\top}
\end{align}
\end{subequations}
}{
\begin{subequations}
\begin{align}
\Xicalbf\Wcalbf(\Gcalbf+\Ccalbf)^{\top}
&=\sum\limits_{i = 1}^{|\Pcal_z|}\Xi_z^{(i)}w_z^{(i)}(\nonlinfunc(\Zcal^{(i)}) + \uvec)^{\top}\\
&\quad+ 
\sum\limits_{i = 1}^{|\Pcal_l|/2}(\Xi_l^{(i)} - \Xi_l^{(i)})w_l^{(i)}(\nonlinfunc(\mn)+\uvec)^{\top}\nonumber\\
%&=\sum\limits_{i = 1}^{|\Pcal_z|}\Xi_z^{(i)}w_z^{(i)}\nonlinfunc(\Zcal^{(i)})^{\top}\\
&=\Xicalbf_z\Wcalbf_z\Gcalbf_z^{\top}
\end{align}
\end{subequations}
}
which, when plugged into~\eqref{eq:thmmainPxy} and~\eqref{eq:thmmainPyy}, along with simplifications of the term $(\Gcalbf + \Ccalbf)\Wcalbf(\Gcalbf + \Ccalbf)^{\top}$ using ~\eqref{eq:diffY} and~\eqref{eq:diffY2}, yields the expression for the joint distribution stated in Proposition~\ref{thm:main}.
%which means that the expressions for the joint distribution simplify to
%\begin{align}
%\Pxy^{\top}&\approx\begin{bmatrix}
%\F_z\Wcalbf_z\Xicalbf_z^{\top}\Lxx^{\top}\\
%\A \Pxx
%\end{bmatrix}\\
%\Pyy&\approx
%\begin{bmatrix}
%(\F-\G)\Wcalbf(\F-\G)^{\top} &
%\F_z\Wcalbf_z\Xicalbf_z^{\top} \Lxx^{\top}\A^{\top}\\
%\A \Lxx \Xicalbf_z\Wcalbf_z\F_z^{\top} & 
%\A \Pxx\A^{\top}
%\end{bmatrix}
%\end{align}
\end{proof}

\begin{proof}[Proof of Corollary~\ref{cor:compgains}]
This remark holds trivially in the case of the SC and UT, but is more complicated when ${\mathbf{N}}\Xicalbf_z$ no longer contains unique columns. To deal with this problem, let the superscript $\Pcal^u_z$ denote the set of points corresponding to the unique columns in ${\mathbf{N}}\Xicalbf_z$, with the corresponding linear parts set to zero and the weights of all multiples of the the same point summed, or mathematically
\begin{align}
\Scal &= \{\nuvec = \N\xivec|(w, \xivec)\in\Pcal_z\},\\ 
\Pcal^u_z &= \Big\{(w^u,\xivec^u)\Big|w^u = \hspace{-20pt}\sum\limits_{(w,\xivec)\in\Pcal_z, \N\xivec = \nuvec,}\hspace{-10pt} w,\;\;\; \xivec^u = \begin{bmatrix}\nuvec\\\Z_{L\times 1}\end{bmatrix}, \nuvec\in\Scal\Big\},\nonumber
\end{align}
such that the set $\Scal$ contains the unique columns of ${\mathbf{N}}\Xicalbf_z$, and the cardinality of the sets satisfy $|\Scal| |\Pcal_z^u| = |\Pcal_z|$. Let
\replaceinpaper{
\begin{align}
w_z^{u,(i)} &= w^{u,(i)}\in\mathbb{R}, &\Xi_z^{u,(i)} &= \xivec^{u,(i)}\in\mathbb{R}^X, &\mathcal{X}_z^{u,(i)} &= \mx + \Lxx\xivec^{u,(i)}\in\mathbb{R}^X,&\forall&(w^{u,(i)},\xivec^{u,(i)})\in\mathcal{P}_z^u,
\end{align}
}{
\begin{align*}
\begin{cases}
w_z^{u,(i)} = w^{u,(i)}\in\mathbb{R}\\
\Xi_z^{u,(i)} = \xivec^{u,(i)}\in\mathbb{R}^X,\\
\mathcal{X}_z^{u,(i)} = \mx \hspace{-2pt}+\hspace{-2pt} \Lxx\xivec^{u,(i)}\in\mathbb{R}^X,
\end{cases}
\forall&(w^{u,(i)},\xivec^{u,(i)})\in\mathcal{P}_z^u,
\end{align*}
}
and, similar to~\eqref{eq:matdef}, define the convenient matrices
\replaceinpaper{
\begin{align}
\wvec_z^u &=\begin{bmatrix}
w_z^{u,(1)} & \cdots & w_z^{u,(|\Pcal_z^u|)}
\end{bmatrix},&
\Xicalbf_z^u &=\begin{bmatrix}
\Xi_z^{u,(1)} & \cdots & \Xi_z^{u,(|\Pcal_z^u|)}
\end{bmatrix},&
\Xcalbf_z^u &=\begin{bmatrix}
\Xcal_z^{u,(1)} & \cdots & \Xcal_z^{u,(|\Pcal_z^u|)}
\end{bmatrix}
\end{align}
}{
\begin{subequations}
\begin{align}
\wvec_z^u &=\begin{bmatrix}
w_z^{u,(1)} & \cdots & w_z^{u,(|\Pcal_z^u|)}
\end{bmatrix},\\
\Xicalbf_z^u &=\begin{bmatrix}
\Xi_z^{u,(1)} & \cdots & \Xi_z^{u,(|\Pcal_z^u|)}
\end{bmatrix},\\
\Xcalbf_z^u &=\begin{bmatrix}
\Xcal_z^{u,(1)} & \cdots & \Xcal_z^{u,(|\Pcal_z^u|)}
\end{bmatrix},
\end{align}
\end{subequations}}
with $\mathcal{Z}_z^{u,(i)} = \mathbf{N}
\mathcal{X}_z^{u,(i)}$. By these definitions, we have
\shortenedEquation{
\begin{align}
\sum\limits_{i = 1}^{|\Pcal_z|}
w^{(i)}_z\nonlinfunc(\Zcal_z^{(i)}) &=
\sum\limits_{j = 1}^{|\Pcal_z^u|}\Bigg(
\sum\limits_{i = 1, \N\Xi_z^{(i)} = \N\Xi_z^{u,(j)}}^{|\Pcal_z|}
w^{(i)}_z\nonlinfunc(\Zcal_z^{(i)})
\Bigg)\\
&=
\sum\limits_{j = 1}^{|\Pcal_z^u|}\Bigg(
\sum\limits_{i = 1, \N\Xi_z^{(i)} = \N\Xi_z^{u,(j)}}^{|\Pcal_z|}
w^{(i)}_z\nonlinfunc(\Zcal_z^{u,(j)})
\Bigg)\\
&= \sum\limits_{j = 1}^{|\Pcal_z^u|}\Bigg(
\sum\limits_{i = 1, \N\Xi_z^{(i)} = \N\Xi_z^{u,(j)}}^{|\Pcal_z|}
w^{(i)}_z
\Bigg)\nonlinfunc(\Zcal_z^{u,(j)})\\
&=\sum\limits_{j = 1}^{|\Pcal_z^u|}
w^{u,(j)}_z\nonlinfunc(\Zcal_z^{u,(j)}).
\end{align}
we also find that
\begin{align}
\sum\limits_{i = 1}^{|\Pcal_z|}\N\Xi_z^{(i)}
w^{(i)}_z\nonlinfunc(\Zcal_z^{(i)})^{\top} &=
\sum\limits_{j = 1}^{|\Pcal_z^u|}\Bigg(
\sum\limits_{i = 1, \N\Xi_z^{(i)} = \N\Xi_z^{u,(j)}}^{|\Pcal_z|}
\N\Xi_z^{(i)}w^{(i)}_z\nonlinfunc(\Zcal_z^{(i)})^{\top}
\Bigg)\\
&=
\sum\limits_{j = 1}^{|\Pcal_z^u|}\Bigg(
\sum\limits_{i = 1, \N\Xi_z^{(i)} = \N\Xi_z^{u,(j)}}^{|\Pcal_z|}
\N\Xi_z^{u,(j)}w^{(i)}_z\nonlinfunc(\Zcal_z^{u,(j)})^{\top}
\Bigg)\\
&= \sum\limits_{j = 1}^{|\Pcal_z^u|}\Bigg(
\sum\limits_{i = 1, \N\Xi_z^{(i)} = \N\Xi_z^{u,(j)}}^{|\Pcal_z|}
w^{(i)}_z
\Bigg)\N\Xi_z^{u,(j)}\nonlinfunc(\Zcal_z^{u,(j)})^{\top}\\\label{eq:53}
&=\sum\limits_{j = 1}^{|\Pcal_z^u|}
\N\Xi_z^{u,(j)}w^{u,(j)}_z\nonlinfunc(\Zcal_z^{u,(j)})^{\top}.
\end{align}
}{
\begin{align}
\sum\limits_{i = 1}^{|\Pcal_z|}
w^{(i)}_z\nonlinfunc(\Zcal_z^{(i)})
&= \sum\limits_{j = 1}^{|\Pcal_z^u|}\Bigg(
\sum\limits_{i = 1, \N\Xi_z^{(i)} = \N\Xi_z^{u,(j)}}^{|\Pcal_z|}
w^{(i)}_z
\Bigg)\nonlinfunc(\Zcal_z^{u,(j)})\nonumber\\\label{eq:53}
&=\sum\limits_{j = 1}^{|\Pcal_z^u|}
w^{u,(j)}_z\nonlinfunc(\Zcal_z^{u,(j)})
\end{align}
and, by similar developments,
\begin{align}
\sum\limits_{i = 1}^{|\Pcal_z|}\N\Xi_z^{(i)}
w^{(i)}_z\nonlinfunc(\Zcal_z^{(i)})^{\top}&=\sum\limits_{j = 1}^{|\Pcal_z^u|}
\N\Xi_z^{u,(j)}w^{u,(j)}_z\nonlinfunc(\Zcal_z^{u,(j)})^{\top}\nonumber
\end{align}
}
Now, since\replaceinpaper{
\begin{equation}\label{eq:strongsym}
    \forall(w^{(i)},\xivec^{(i)})\in\Pcal_z, \exists(w^{(j)},\xivec^{(j)})\in\Pcal_z\quad\text{such that}\quad(w^{(i)}, \N\xivec^{(i)},\bar{\N}\xivec^{(i)}) = (w^{(j)},{\N}\xivec^{(j)},-\bar{\N}\xivec^{(j)}),
\end{equation}}{
$\forall(w^{(i)},\xivec^{(i)})\in\Pcal_z, \exists(w^{(j)},\xivec^{(j)})\in\Pcal_z$ such that $(w^{(i)}, \N\xivec^{(i)},\bar{\N}\xivec^{(i)}) = (w^{(j)},{\N}\xivec^{(j)},-\bar{\N}\xivec^{(j)})$}
by assumption, choose an ordering of the columns in $\Xicalbf_z$ where
\begin{equation}\label{eq:55}
\begin{bmatrix}
\N\\\bar{\N}
\end{bmatrix}\Xi_z^{(i)} = 
\begin{bmatrix}
\N\\-\bar{\N}
\end{bmatrix}\Xi_z^{(|\Pcal_z|/2+i)}, \;\; \forall i = 1,...,|\Pcal_z|/2.
\end{equation}
then $w^{(i)} = w^{(|\Pcal_z|/2+i)}$, and
\begin{equation}\label{eq:56}
\Zcal^{(i)} = \N\Xcal^{(i)} \removeinpaper{= \N(\mx + \Lxx \xivec^{(i)}) = \N(\mx + \Lxx \xivec^{(|\Pcal_z|/2+i)})}= \N\Xcal^{(|\Pcal_z|/2+i)} = \Zcal^{(|\Pcal_z|/2+i)}
\end{equation}
By examination of the product $\Xicalbf\Wcalbf(\Gcalbf+\Ccalbf)^{\top}$, and using~\eqref{eq:53},~\eqref{eq:55} and~\eqref{eq:56}, we find that
\shortenedEquation{
\begin{align}
\Xicalbf\Wcalbf(\Gcalbf+\Ccalbf)^{\top} &=\sum\limits_{i = 1}^{|\Pcal_z|}\Xi_z^{(i)}w_z^{(i)}\nonlinfunc(\Zcal^{(i)})^{\top}\\
&=\sum\limits_{i = 1}^{|\Pcal_z|}
\begin{bmatrix}
\N\\\bar{\N}
\end{bmatrix}
\Xi_z^{(i)}w_z^{(i)}\nonlinfunc(\Zcal^{(i)})^{\top}\\
&=\sum\limits_{i = 1}^{|\Pcal_z|/2}
\begin{bmatrix}
2\N\\\bar{\N}-\bar{\N}
\end{bmatrix}
\Xi_z^{(i)}w_z^{(i)}\nonlinfunc(\Zcal^{(i)})^{\top}\\
&=\sum\limits_{i = 1}^{|\Pcal_z|}
\begin{bmatrix}
\N\\\Z_{L\times X}
\end{bmatrix}
\Xi_z^{(i)}w_z^{(i)}\nonlinfunc(\Zcal^{(i)})^{\top}\\
&=
\sum\limits_{i = 1}^{|\Pcal_z^u|}\Xi_z^{u,(i)}w_z^{u,(i)}\nonlinfunc(\Zcal^{u,(i)})^{\top}\\
&=\Xicalbf_z^u\Wcalbf_z^u\Gcalbf_z^u
\end{align}
}{
\begin{align}
\Xicalbf\Wcalbf(\Gcalbf\hspace{-1pt}+\hspace{-1pt}\Ccalbf)^{\top}\hspace{-4pt}=\hspace{-2pt}\sum\limits_{i = 1}^{|\Pcal_z|}\Xi_z^{(i)}w_z^{(i)}\nonlinfunc(\Zcal^{(i)})^{\top}\hspace{-4pt}=\hspace{-1pt}\Xicalbf_z^u\Wcalbf_z^u(\Gcalbf_z^u)^{\top}\hspace{-2pt}
\end{align}
}
and trivially
\begin{equation}
(\Gcalbf_z + \Ccalbf_z)\Wcalbf_z(\Gcalbf_z + \Ccalbf_z)^{\top} \removeinpaper{+ w_{cl}\uvec_{\lvec}\uvec_{\lvec}^{\top}}\hspace{-2pt}+\hspace{-1pt} (\Gcalbf_z^u + \Ccalbf_z^u)\Wcalbf_z^u(\Gcalbf_z^u  + \Ccalbf_z^u)^{\top} \removeinpaper{+ w_{cl}\uvec_{\lvec}\uvec_{\lvec}^{\top}}\hspace{-3pt},
\end{equation}
\begin{subequations}
where
\begin{align}
w_{cl} &= 1 - \wvec_z \mathbf{1}_{|\Pcal_z|}\\
\Gcalbf_z &= \begin{bmatrix}
\nonlinfunc(\Zcal^{u,(1)}) & \cdots & \nonlinfunc(\Zcal^{u,(\Pcal_z^u|)})
\end{bmatrix}\\
\uvec &= -w_{cl}\nonlinfunc(\mn) - \Gcalbf_z\wvec_z^{\top}\\
\uvec_{\lvec}&= \nonlinfunc(\mn) + \uvec \\
\Ccalbf_z^u &= \uvec\mathbf{1}_{|\Pcal_z^u|}^{\top}
\end{align}
\end{subequations}
Note here that
\begin{align}
w_{cl} &= 1 - \wvec_z \mathbf{1}_{|\Pcal_z|}= 1 - \wvec_z^u \mathbf{1}_{|\Pcal_z^u|}\\ \uvec  &= -w_{cl}\nonlinfunc(\mn) - \Gcalbf_z\wvec_z^{\top} =  -w_{cl}\nonlinfunc(\mn) - \Gcalbf_z^u(\wvec_z^u)^{\top}\nonumber
\end{align}
Thus, Proposition~\ref{thm:main} yields identical results when evaluated with $\Pcal_z$ and $\Pcal_z^u$. And since $|\Pcal_z^u| = C(Z)$ the total number of function evaluation scales with $C(Z)$ instead of $C(X)$, provided the symmetry assumption in Corollary~\ref{cor:compgains} holds. 
\end{proof}

\begin{proof}[Proof of Corollary~\ref{cor:A}]
If $\bar{\mathbf{N}}\Xicalbf_z = \Z$,
\begin{equation}
\Lxx \Xicalbf_z = \begin{bmatrix}
\Lnn&\Z\\
\Lln&\Lll
\end{bmatrix}
\begin{bmatrix}
\mathbf{N}\\
\bar{\mathbf{N}}
\end{bmatrix}\Xicalbf_z 
\removeinpaper{= \begin{bmatrix}
\Lnn&\Z\\
\Lln&\Lll
\end{bmatrix}
\begin{bmatrix}
\mathbf{N}\Xicalbf_z\\
\bar{\mathbf{N}}\Xicalbf_z
\end{bmatrix} = 
\begin{bmatrix}
\Lnn\mathbf{N}\Xicalbf_z\\
\Lln\mathbf{N}\Xicalbf_z + \Lll\bar{\mathbf{N}}\Xicalbf_z
\end{bmatrix} }= 
\begin{bmatrix}
\Lnn\\
\Lln
\end{bmatrix}\mathbf{N}\Xicalbf_z.
\vspace*{-7pt}\end{equation}
\end{proof}

\begin{proof}[Proof of Corollary~\ref{cor:otherstruct}]
Here, we can let the number of rows of $\A_1$ and $\A_2$ in~\eqref{eq:otherstruct} be $N_1$ and $N_2$ respectively. Let $\A^{\top} = [\A_1^{\top}, \A_2^{\top}]$ and define a function with output $\ovec$, and a map
\begin{equation}
\ovec
\hspace{-1pt}=\hspace{-1pt}
\begin{bmatrix}
\gvec(\zvec)\\\A\xvec
\end{bmatrix}
%\hspace{-1pt}=\hspace{-1pt}\begin{bmatrix}\gvec(\zvec)\\\A_1\xvec\\\A_2\xvec\end{bmatrix}
, \;\; \M
\hspace{-1pt}=\hspace{-1pt}
\begin{bmatrix}
\I_{N_1} & \I_{N_1} & \Z \\
\Z & \Z & \I_{N_2}
\end{bmatrix}\hspace{-1pt}.
\end{equation}
such that with the structure of the function defining $\yvec$ in~\eqref{eq:otherstruct}, we have that $\yvec = \M\ovec$. Then, the joint distribution $p(\xvec,\ovec)$ can simply be computed using Proposition~\ref{thm:main} to find $\boldsymbol{m}^{\ovec},\PP^{\xvec\ovec},\PP^{\ovec\ovec}$, followed by computation of the joint distribution $p(\xvec,\yvec)$ in terms of $\boldsymbol{m}^{\yvec},\PP^{\xvec\yvec},\PP^{\yvec\yvec}$ where then
\begin{equation}
\my \hspace{-1pt}=\hspace{-1pt}\M\boldsymbol{m}^{\ovec}, \Pxy\hspace{-1pt}=\hspace{-1pt}\PP^{\xvec\ovec}\M^{\top}, \Pyy\hspace{-1pt}=\hspace{-1pt}\M\PP^{\ovec\ovec}\M^{\top}.
\vspace*{-5pt}\end{equation}
\vspace*{-7pt}\end{proof}

\bibliography{main_arxiv_final_version}{}

\begin{thebibliography}{10}
\providecommand{\url}[1]{#1}
\csname url@samestyle\endcsname
\providecommand{\newblock}{\relax}
\providecommand{\bibinfo}[2]{#2}
\providecommand{\BIBentrySTDinterwordspacing}{\spaceskip=0pt\relax}
\providecommand{\BIBentryALTinterwordstretchfactor}{4}
\providecommand{\BIBentryALTinterwordspacing}{\spaceskip=\fontdimen2\font plus
\BIBentryALTinterwordstretchfactor\fontdimen3\font minus
  \fontdimen4\font\relax}
\providecommand{\BIBforeignlanguage}[2]{{%
\expandafter\ifx\csname l@#1\endcsname\relax
\typeout{** WARNING: IEEEtran.bst: No hyphenation pattern has been}%
\typeout{** loaded for the language `#1'. Using the pattern for}%
\typeout{** the default language instead.}%
\else
\language=\csname l@#1\endcsname
\fi
#2}}
\providecommand{\BIBdecl}{\relax}
\BIBdecl

\bibitem{steinbring2014lrkf}
J.~Steinbring and U.~D. Hanebeck, ``{LRKF} revisited: The smart sampling
  {Kalman} filter ({S2KF}),'' \emph{Journal of Advances in Information Fusion},
  vol.~9, no.~2, pp. 106--123, 2014.

\bibitem{kurz2017linear}
G.~Kurz and U.~D. Hanebeck, ``Linear regression {Kalman} filtering based on
  hyperspherical deterministic sampling,'' in \emph{2017 IEEE 56th Annual
  Conference on Decision and Control (CDC)}.\hskip 1em plus 0.5em minus
  0.4em\relax IEEE, 2017, pp. 977--983.

\bibitem{gordon1993novel}
N.~J. Gordon, D.~J. Salmond, and A.~F. Smith, ``Novel approach to
  {nonlinear/non-Gaussian Bayesian} state estimation,'' in \emph{IEE
  proceedings (radar and signal processing)}, vol. 140, no.~2, 1993, pp.
  107--113.

\bibitem{smith2013sequential}
A.~Smith, \emph{{Sequential Monte Carlo} methods in practice}.\hskip 1em plus
  0.5em minus 0.4em\relax Springer Science \& Business Media, 2013.

\bibitem{schon2005marginalized}
T.~Schon, F.~Gustafsson, and P.-J. Nordlund, ``Marginalized particle filters
  for mixed linear/nonlinear state-space models,'' \emph{IEEE Transactions on
  signal processing}, vol.~53, no.~7, pp. 2279--2289, 2005.

\bibitem{schon2006marginalized}
T.~B. Schon, R.~Karlsson, and F.~Gustafsson, ``The marginalized particle filter
  in practice,'' in \emph{2006 IEEE Aerospace Conference}.\hskip 1em plus 0.5em
  minus 0.4em\relax IEEE, 2006, pp. 11--pp.

\bibitem{sarkka2013bayesian}
S.~S{\"a}rkk{\"a}, \emph{Bayesian filtering and smoothing}.\hskip 1em plus
  0.5em minus 0.4em\relax Cambridge University Press, 2013, vol.~3.

\bibitem{arasaratnam2009cubature}
I.~Arasaratnam, ``Cubature {Kalman} filtering theory \& applications,'' Ph.D.
  dissertation, McMaster University, 2009.

\bibitem{jia2013high}
B.~Jia, M.~Xin, and Y.~Cheng, ``{High-degree cubature Kalman filter},''
  \emph{Automatica}, vol.~49, no.~2, pp. 510--518, 2013.

\bibitem{julier1997new}
S.~J. Julier and J.~K. Uhlmann, ``New extension of the {Kalman} filter to
  nonlinear systems,'' in \emph{Signal processing, sensor fusion, and target
  recognition VI}, vol. 3068.\hskip 1em plus 0.5em minus 0.4em\relax
  International Society for Optics and Photonics, 1997, pp. 182--194.

\bibitem{wan2001unscented}
E.~A. Wan, R.~Van Der~Merwe, and S.~Haykin, ``The unscented {K}alman filter,''
  \emph{Kalman filtering and neural networks}, vol.~5, no. 2007, pp. 221--280,
  2001.

\bibitem{dunik2011development}
J.~Dun{\'\i}k, O.~Straka, and M.~{\v{S}}imandl, ``The development of a
  randomised unscented {Kalman} filter,'' \emph{IFAC Proceedings Volumes},
  vol.~44, no.~1, pp. 8--13, 2011.

\bibitem{dunik2013stochastic}
------, ``Stochastic integration filter,'' \emph{IEEE Transactions on Automatic
  Control}, vol.~58, no.~6, pp. 1561--1566, 2013.

\bibitem{morelande2007unscented}
M.~R. Morelande and B.~Moran, ``An unscented transformation for conditionally
  linear models,'' in \emph{2007 IEEE International Conference on Acoustics,
  Speech and Signal Processing-ICASSP}, vol.~3.\hskip 1em plus 0.5em minus
  0.4em\relax IEEE, 2007, pp. III--1417.

\bibitem{beutler2009gaussian}
F.~Beutler, M.~F. Huber, and U.~D. Hanebeck, ``Gaussian filtering using state
  decomposition methods,'' in \emph{2009 12th International Conference on
  Information Fusion}.\hskip 1em plus 0.5em minus 0.4em\relax IEEE, 2009, pp.
  579--586.

\bibitem{ito2000gaussian}
K.~Ito, ``Gaussian filter for nonlinear filtering problems,'' in
  \emph{Proceedings of the 39th IEEE Conference on Decision and Control (Cat.
  No. 00CH37187)}, vol.~2.\hskip 1em plus 0.5em minus 0.4em\relax IEEE, 2000,
  pp. 1218--1223.

\bibitem{golub1969calculation}
G.~H. Golub and J.~H. Welsch, ``Calculation of {Gauss} quadrature rules,''
  \emph{Mathematics of computation}, vol.~23, no. 106, pp. 221--230, 1969.

\bibitem{singer1970estimating}
R.~A. Singer, ``Estimating optimal tracking filter performance for manned
  maneuvering targets,'' \emph{IEEE Transactions on Aerospace and Electronic
  Systems}, no.~4, pp. 473--483, 1970.

\end{thebibliography}
\bibliographystyle{IEEEtran}

\end{document}